\documentclass[a4paper]{amsart}
\usepackage{amssymb}
\usepackage[many]{tcolorbox}
\usepackage{graphicx} 
\usepackage{mathrsfs}
\usepackage{wasysym}
\usepackage{braket}
\usepackage{hyperref}
\usepackage{mathtools}
\usepackage{microtype}
\usepackage{enumitem}
\usepackage{xcolor}
\usepackage{mdwtab}

\usepackage{complexity}
\usepackage{slashed}
\usepackage[store-sets-label]{keytheorems}
\usepackage{zref-clever}
\usepackage[all]{xy}
\zcsetup{nameinlink,cap}
\usepackage{todonotes}
\newcommand{\ctimes}{\mathbin{\normalfont\Circle}}
\DeclareMathOperator{\atp}{\rm atp}
\DeclareMathOperator{\otp}{\rm otp}
\DeclareMathOperator{\sgn}{\rm sgn}
\DeclareMathOperator{\Gaif}{\rm Gaif}
\DeclareMathOperator{\dom}{\rm Dom}
\DeclareMathOperator{\im}{\rm Im}

\DeclareMathOperator{\Red}{\rm Red}
\newcommand{\trans}[3][R]{\textrm{W}_{\!#2,#3}^{#1}}
\newcommand{\strans}[3][R]{\overline{\textrm{W}}_{\!#2,#3}^{#1}}
\newcommand{\is}{\mathrel{:=}}
\newcommand{\overlay}{\mathbin{\vee}}
\newcommand{\simor}{\mathrel{\sim_{\rm or}}}
\newcommand{\sima}{\mathrel{\sim_\alpha}}
\newcommand{\oak}{\overrightarrow{\mathscr A_k}}
\newcommand{\obk}{\overrightarrow{\mathscr B_k}}
\newcommand{\osk}{\overrightarrow{\mathscr S_k}}
\newcommand{\oC}{\overrightarrow{\mathscr C}}
\newcommand{\oD}{\overrightarrow{\mathscr D}}
\newcommand{\oE}{\overrightarrow{\mathscr E}}
\newcommand{\oF}{\overrightarrow{\mathscr F}}

\newkeytheoremstyle{tcbp}{tcolorbox-no-titlebar={interior hidden, breakable,before skip=10pt,after skip=10pt,	outer arc=0pt, 	arc=0pt, blanker,
		borderline west={1pt}{0pt}{blue},
		top=3pt,
		right=3pt,
		left=5pt,
		bottom=3pt},
}
\newkeytheoremstyle{tcbthm}{tcolorbox-no-titlebar={interior hidden, breakable,before skip=10pt,after skip=10pt,	outer arc=0pt, 	arc=0pt, blanker,
		borderline west={1pt}{0pt}{red},
		top=3pt,
		right=3pt,
		left=5pt,
		bottom=3pt},
}
\newkeytheoremstyle{tcbethm}{tcolorbox-no-titlebar={interior hidden, breakable,before skip=10pt,after skip=10pt,	outer arc=0pt, 	arc=0pt, blanker,
		borderline west={1pt}{0pt}{gray},
		top=3pt,
		right=3pt,
		left=5pt,
		bottom=3pt},
}
\newkeytheoremstyle{tcblem}{tcolorbox-no-titlebar={breakable,before skip=10pt,after skip=10pt, blanker, borderline west={1pt}{0pt}{orange},
		top=3pt, right=3pt, left=5pt, bottom=3pt},
}
\newkeytheoremstyle{tcbslem}{tcolorbox-no-titlebar={breakable,before skip=10pt,after skip=10pt,  blanker, borderline={1pt}{0pt}{orange},
		top=13pt, right=13pt, left=15pt, bottom=13pt},
}
\newkeytheoremstyle{tcbpb}{tcolorbox-no-titlebar={breakable,before skip=10pt,after skip=10pt,	blanker,
		borderline west={1pt}{0pt}{cyan},
		top=3pt,
		right=3pt,
		left=5pt,
		bottom=3pt},
}
\newkeytheoremstyle{tcbrk}{tcolorbox-no-titlebar={breakable,before skip=10pt,after skip=10pt,	blanker,
		borderline west={2pt}{0pt}{gray},
		top=3pt,
		right=3pt,
		left=5pt,
		bottom=3pt},
}
\newkeytheoremstyle{tcbclaim}{tcolorbox-no-titlebar={breakable,before skip=10pt,after skip=10pt,	blanker,
		borderline west={1pt}{0pt}{yellow},
		top=3pt,
		right=3pt,
		left=5pt,
		bottom=3pt},
}

\newkeytheorem{theorem}[style=tcbthm,name=Theorem,parent=section]
\newkeytheorem{etheorem}[style=tcbethm,name=Theorem,sibling=theorem]
\newkeytheorem{definition}
[style=tcbp,name=Definition,sibling=theorem]
\newkeytheorem{lemma}[style=tcblem,name=Lemma,sibling=theorem]
\newkeytheorem{slemma}[style=tcbslem,name=Lemma,sibling=lemma]
\newkeytheorem{corollary}[style=tcbthm,name=Corollary,sibling=theorem]
\newkeytheorem{problem}[style=tcbpb,name=Problem]
\newkeytheorem{conjecture}[style=tcbpb,name=Conjecture]
\newkeytheorem{remark}[style=tcbrk,name=Remark,numbered=false]
\newkeytheorem{fact}[name=Fact,sibling=theorem]
\newkeytheorem{claim}[style=tcbclaim,name=Claim,sibling=theorem]
\zcRefTypeSetup{claim}{
	Name-sg = Claim ,
	name-sg = claim ,
	Name-pl = Claims ,
	name-pl = claims ,
}
\zcRefTypeSetup{etheorem}{
	Name-sg = Theorem ,
	name-sg = theorem ,
	Name-pl = Theorems ,
	name-pl = theorems ,
}
\thanks{\ERCagreement}
\author[H. Buffière]{Hector Buffière}
\address{Université Paris Cité, CNRS, IRIF, Paris, France  \and Centre d'Analyse et de Mathé\-ma\-tique Sociales CNRS UMR 8557, France}
\email{buffiere@irif.fr}
\author[Y. Lin]{Yuquan Lin}
\address{Southeast University, Nanjing, Jiangsu, China \and Centre d'Analyse et de Mathé\-ma\-tique Sociales CNRS UMR 8557, France.}
\email{yqlin@seu.edu.cn}
\author[P. Ossona de Mendez]{Patrice Ossona de Mendez}
\address{Centre d'Analyse et de Mathématique Sociales CNRS UMR 8557, France \and Computer Science Institute of Charles University (IUUK), Praha, Czech Republic}
\email{pom@ehess.fr}

\date{\today}
\newcommand{\ERCagreement}{
	{\footnotesize
		The second author is supported by the China Scholarship Council (CSC)	and  SEU Innovation Capability Enhancement Plan for Doctoral Students (CXJH\_SEU 24119).}\\[4pt]
	\noindent\begin{minipage}{.73\textwidth}
		\footnotesize
		This paper is part of a project that has received funding from the European Research Council (ERC) under the European Union's Horizon 2020 research and innovation program (grant agreement No 810115 -- {\sc Dynasnet}).
	\end{minipage}\hfill
	\begin{minipage}{.25\textwidth}
		\phantom{.}\hfill\includegraphics[height=13mm]{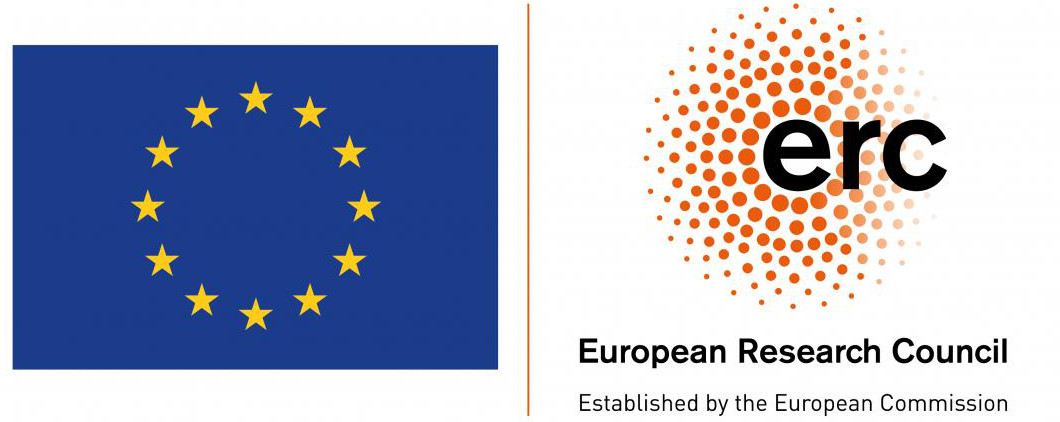}\hfill\phantom{.}
	\end{minipage}
}

\newenvironment{clproof}{ \trivlist
	\item[\hskip\labelsep
	\emph{Proof of the claim}.]\ignorespaces
}{\hfill$\vartriangleleft$\vspace{10pt}\par}

\date{\today}

\title[Monadic dependence from reducts]{Monadic dependence from reducts, and applications to twin-width of oriented graphs}

\begin{document}
	\begin{abstract}
    We study monadic dependence of binary relational structures including at least one antisymmetric relation. 

    Our cornerstone result gives sufficient conditions for proving that a structure is monadically dependent
    by only considering some of its reducts, assuming they are structurally well-behaved and compatible enough. 
    As an application, we consider some reorientation rules preserving monadic dependence of binary structures, as well as replacement of one antisymmetric relation with bounded independence number by another.
    Then, we apply our main (technical) result to the study of twin-width in two ways.

    First, generalizing the fact that twin-width boundedness is equivalent to being expandable by a linear
    order into a monadically dependent class, we prove that it is also equivalent to being expandable by
    an oriented graph with bounded independence number (for instance by a poset with bounded width or by a tournament), and that
    {\sf FO}-model checking is fixed-parameter tractable on such an expansion.

    Second, we show delineation by twin-width for some new classes of oriented graphs, including oriented
    split graphs and local tournaments. In all these cases, we also obtain fixed-parameter tractability of  {\sf FO}-model
    checking.
	\end{abstract}
	\maketitle

\section{Introduction}
\subsection{Background}
Dependence (or NIP) is a strong model theoretical dividing line~\cite{NIP}, whose importance in the structural and algorithmic study of hereditary classes of graphs is witnessed, for example, by the following conjecture.
\begin{conjecture}
\label{conj:mc}
    Under the standard assumption $\FPT\neq\AW[\ast]$, 
    first-order model checking is fixed parameter tractable on a hereditary class of graphs $\mathscr C$ if and only if $\mathscr C$ is dependent.
\end{conjecture}
Progress toward \zcref{conj:mc} includes a proof that \FO-model checking is {\FPT} on every hereditary stable class of graphs \cite{MCST}, on every hereditary dependent class of ordered graphs \cite{tww4}, classes with bounded twin-width \cite{twin-width1} (under the assumption of the availability of a contraction sequence with bounded width), and classes with bounded merge-width \cite{dreier2025merge} (under the assumption of the availability of a merge sequence with bounded width).

In the context of hereditary classes, which is standard in graph theory, dependence is equivalent to the stronger property of monadic dependence and allows an existential characterization \cite{braunfeld2022existential}, which simplifies its structural study. Indeed, some purely combinatorial characterizations of hereditary independent classes of graphs have been established \cite{dreier2024flipbreakability,svm2024}.
The characterization given in \cite{dreier2024flipbreakability} was obtained by first describing general forbidden configurations in monadically independent binary relational structures (called \emph{transformers}), which is an essential step toward characterization theorems for non-symmetric binary structures (see \cite{tows_arxiv} for instance).

The recently introduced  \emph{twin-width} invariant \cite{twin-width1} gets here a special importance for two main reasons. First, twin-width is strongly tied to monadic dependence for classes of ordered binary structures, as established by the following theorem.
\begin{etheorem}[\cite{tww4}]
\label{thm:tww4}
A class of binary structures has bounded twin-width if and only if it can be expanded into a monadically dependent class of linearly ordered binary structures.
\end{etheorem}
The second reason is algorithmic.
The twin-width invariant is built on the notion of a \emph{contraction sequence}, which can intuitively be seen as a way to iteratively identify near-twins in a graph, while keeping a cumulative error bounded. Assuming a contraction sequence is given, first-order model checking can be performed in linear time (parametrized by the width of the contraction sequence).
However, no general polynomial-time algorithm is known yet that is able to construct a contraction sequence with width $f(t)$ for graphs with twin-width at most $t$ in polynomial time (parametrized by $t$). Nevertheless, such contraction sequences can be efficiently computed in some special cases, including 
ordered graphs~\cite{tww5}, and tournaments \cite{geniet2025orderlogictwinwidthtournaments}. 

These examples  witness a particular property called \emph{delineation}:
A graph class $\mathscr C$ is  \emph{delineated} by twin-width (or simply,
delineated) if every subclass $\mathscr D$ of $\mathscr C$ has bounded twin-width
if and only if it is monadically dependent \cite{bonnet_et_al:LIPIcs.IPEC.2022.9}. More: for a class $\mathscr C$ of ordered graphs (resp. of tournaments), either the class $\mathscr C$ is independent and \FO-model checking is \AW[$\ast$]-hard, or the class $\mathscr C$ has bounded twin-width and \FO-model checking is \FPT. This behavior
applies as well to some other classes delineated, such as circle graphs~\cite{gajarskynew}, $H$-graphs (for any fixed
forest $H$) \cite{bonomo2025non}, and particular geometric graphs~\cite{geniet2025first}. 
On the other hand, the class of cubic graphs is not delineated by twin-width 
(as it is monadically dependent but has unbounded twin-width \cite{twin-width1}).

As noticed above, an important example of delineation is provided by the class of tournaments.
\begin{etheorem}[\cite{geniet2025orderlogictwinwidthtournaments}]
\label{thm:tourn}
The class of tournaments is delineated by twin-width. In other words,
a class of tournaments has bounded twin-width if and only if it is monadically dependent.    
\end{etheorem}

Moreover, \zcref{thm:tourn} extends to classes of oriented graphs with bounded independence number, which shows that 
a sufficiently dense ordered graph is sufficient to deduce the delineation.

\begin{etheorem}[\cite{geniet2025orderlogictwinwidthtournaments}]
\label{thm:alpha}
For every integer $k$,
the class of oriented graphs with independence number at most $k$ is delineated by twin-width.

Moreover, \FO-model checking is {\FPT} on every monadically dependent class of oriented graphs with independence number at most $k$.
\end{etheorem}

Note that a linear order is nothing but a transitive tournament, and it is thus natural to wonder to which extent the assumption on the expansion appearing in the statement of \zcref{thm:tww4} could be
weakened.

\begin{problem}
\label{pb:char}
    Does a class of binary structures have bounded twin-width if and only if it can be expanded into a monadically dependent class of  binary structures by addition of a tournament? of an oriented graph with independence number at most $k$ (for some fixed $k$)? If yes, is \FO-model checking {\FPT} on such an expansion?
\end{problem}

Hereditary classes delineated by twin-width seem to enjoy a nice characterization of when hereditary subclasses have fixed parameter tractable \FO-model checking: in known cases, as in the case of linearly ordered binary structures, either a subclass is not (monadically) dependent and \FO-model checking is $\AW[\ast]$-hard, or it is (monadically) dependent and \FO-model checking is \FPT. Note that this does not directly follow from delineation, as fixed parameter tractability of \FO-model checking for general classes with bounded twin-width is known only when contraction sequences can be computed in polynomial time (such as for monadically dependent classes of ordered binary structures \cite{tww5}, of oriented graphs with bounded independence number \cite{geniet2025orderlogictwinwidthtournaments}, or of special geometric graphs \cite{geniet2025first}).

\begin{problem}
\label{pb:del}
    Find more general delineated classes of oriented graphs than tournaments and orientations of graphs with bounded independence number and check whether \FO-model checking is {\FPT} on such classes.
\end{problem}

On the other hand, \zcref{thm:tourn,thm:alpha} suggests to study how the reorientation of an antisymmetric binary relation behaves with respect to monadic dependence and how the classes of oriented graphs differ from classes of unoriented graphs from the point of view of delineation.

\subsection{Our results}

Our main result establishes that the monadic dependence of a class of relational structures can, under some structural and compatibility conditions, be deduced from the monadic dependence of some of its reducts.
Precisely, we consider a binary relational signature $\sigma$, including two antisymmetric binary relations $R$ and $R'$ and wonder when the monadic dependence of a class $\mathscr C$ of $\sigma$-structures can be deduced from the monadic dependence of the class $\mathscr C^{\sigma\setminus\{R\}}$ (where the relation $R$ is forgotten) and the class $\mathscr C^{\{R,R'\}}$ (where only the relations $R$ and $R'$ are kept).
We prove that this is the case if the $R$-reduct $\mathbf M^R$ and $R'$-reduct $\mathbf M^{R'}$ of every structure $\mathbf M\in\mathscr C$ are oriented graphs in a class $\obk$, which is the class of all the orientations of graphs  excluding the biclique $K_{k,k}$ and the complete split graph $S_{k,k}$, and if some patterns in the $R$-reduct cannot coincide with some patterns in the $R'$-reduct
(Property $\Pi_{R,R'}^k$; see \zcref{fig:pis}).
Formally, we establish the following theorem.

\begin{figure}[ht]
    \centering
    \includegraphics[width=.75\linewidth]{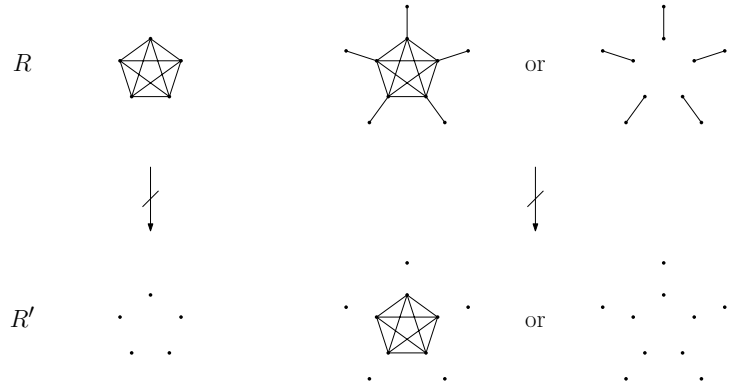}
    \caption{The forbidden modifications from $\mathbf M^R$ to $\mathbf M^{R'}$.}
    \label{fig:pis}
\end{figure}

\getkeytheorem{thm:main}


We derive several applications of this result.
To begin with, we consider when an antisymmetric relation (defining an oriented graph reduct in $\obk$) can be ``reoriented'' while preserving monadic dependence. The formal statement of this result involves sequences of binary structures instead of classes and an equivalence relation $\simor$ on orientations. We quote it here for completeness, but leave its detailed explanation to \zcref{sec:reorient}.
\getkeytheorem{thm:reorient}
Then, in \zcref{sec:replace}, we follow similar arguments to study when an antisymmetric relation with bounded independence number can be replaced by another antisymmetric relation with bounded independence number. The formal statement of this result then involves an equivalence relation $\sima$ on oriented graphs with bounded independence number. Again, we quote the obtained theorem (where $\osk$ denotes the class of oriented graphs with independence number less than $k$) but refer to \zcref{sec:replace} for explanations.
\getkeytheorem{thm:calpha}

Then, we consider some applications to the study of twin-width of binary structures.
First, we give the following positive answer to \zcref{pb:char}.

\getkeytheorem{thm:nalpha}

In particular, intermediate between \zcref{thm:tww4} and \zcref{thm:nalpha}, we have the next characterizations.
\getkeytheorem{cor:tourn}

\getkeytheorem{cor:width}

We then address the general \zcref{pb:del}.
\zcref{thm:alpha} provides an example of a delineated class of oriented graphs such that the underlying class of unoriented graphs is not delineated. We extend \zcref{thm:alpha} by proving the class of all orientations of the graphs that have bounded independence number after removal of an independent set is delineated.

\getkeytheorem{thm:delin}

In particular, while the class of split graphs is not delineated (as it includes a class transduction-equivalent to the class of bipartite subcubic graphs; see \cite{bonnet_et_al:LIPIcs.IPEC.2022.9}), the situation changes when we consider orientations.

\getkeytheorem{cor:split}

Let us mention some algorithmic consequence of this corollary.
A \emph{quasi-kernel} of a digraph $\mathbf D$ is an independent set $K$ such that every 
vertex can reach $K$ by a directed path of  length at most two. 
\textsc{Quasi-Kernel} is the computational problem, whose input is a digraph and an integer $k$, that consists in deciding whether the input digraph has a quasi-kernel of size at most $k$.
It is known that \textsc{Quasi-Kernel} is $\W[2]$-complete  for split digraphs when $k$ is the parameter \cite{langlois2025quasi}.
However, it becomes {\FPT} (with parameter $k$) on any monadically dependent class of orientations of split graphs, as a consequence of \zcref{cor:split}. Note that (under the assumption $\FPT\neq\W[2]$) the  frontier between $\W[2]$-hardness and fixed parameter tractability of \textsc{Quasi-Kernel}  still needs to be made precise, even for classes of split digraphs.

\begin{problem}
Is \textsc{Quasi-Kernel} $\W[2]$-complete on the class of all orientations of split-graphs? Is it {\FPT} on monadically dependent classes of split digraphs?
\end{problem}

Note that the class of split digraphs is not delineated by twin-width (as the class of split graphs is not delineated by twin-width). However, \zcref{conj:mc} would imply that \textsc{Quasi-Kernel} is {\FPT} on monadically dependent classes of split digraphs.
\medskip

It was known that the class of circular arc graphs is delineated \cite{geniet2025first}. We prove here that this
is also the case of local tournaments, which 
are particular orientations of proper circular arc graphs. This extends \zcref{thm:tourn}, which states that the class of tournaments is delineated \cite{geniet2025orderlogictwinwidthtournaments}.

\getkeytheorem{thm:ltourn}

\subsection{Structure of the paper}
In \zcref{sec:prelim}, we recall some basic notions from
graph theory, order theory, and model theory, and introduce some definitions and notations used throughout this paper.
\zcref{sec:main} is dedicated to the proof of \zcref{thm:main}, which is our main result.
Applications of \zcref{thm:main} are given in \zcref{sec:app}, which include \zcref{thm:reorient,thm:calpha,thm:nalpha,thm:delin,thm:ltourn}. 
\newpage

\section{Definitions and notations}
\label{sec:prelim}
\subsection{Structures}
A \emph{relational signature} is a set $\sigma$ of relation symbols, each with an arity. A \emph{$\sigma$-structure} $\mathbf M$
is a set $M$ (the \emph{domain} of $\mathbf M$) and an 
interpretation of the relation symbols. (When $\sigma$ contains a single relation $R$, we use the term of $R$-structure instead of $\{R\}$-structure.)
For each $R\in\sigma$ with arity $k$, we denote by $R(\mathbf M)$ the set of $k$-tuples $(v_1,\dots,v_k)$ of elements of $M$ such that 
$\mathbf M\models R(v_1,\dots,v_k)$.

The \emph{atomic type} of a tuple $\bar v$ of elements of the domain of a $\sigma$-structure $\mathbf M$ is the set $\atp(\bar v)$ of all 
formulas $\phi(\bar x)$ satisfied by $\bar v$ in $\mathbf M$, where each
$\phi(\bar x)$ has the form $R(x_{i_1},\dots,x_{i_k})$ or $\neg R(x_{i_1},\dots,x_{i_k})$ with $R\in\sigma\cup\{=\}$ and $1\leq i_1,\dots,i_k\leq |\bar v|$.

The \emph{Gaifman graph} $\Gaif(\mathbf M)$ is a graph with same domain as $\mathbf M$, where two distinct vertices are adjacent if they belong together to some relation of $\mathbf M$, that is:
\[
\Gaif(\mathbf M)\models E(u,v)\qquad\iff\qquad
\exists R\in\sigma, \exists\bar z\subset \mathbf M,\ u,v\in\bar z\text{ and }\mathbf M\models R(\bar z).
\]

In this paper we only consider binary relational structures, meaning that each relation symbol has arity at most two.

Recall that a binary relation $R$ is 
\begin{itemize}
    \item  \emph{symmetric} in $\mathbf M$ if 
$\mathbf M\models R(x,y)\rightarrow R(y,x)$;
\item \emph{antisymmetric} in $\mathbf M$ if 
$\mathbf M\models (R(x,y)\wedge R(y,x))\rightarrow (x=y)$;
\item \emph{total} in $\mathbf M$ if
$\mathbf M\models (\neg R(x,y)\wedge \neg R(y,x))\rightarrow (x=y)$;
\item \emph{transitive} in $\mathbf M$ if
$\mathbf M\models (R(x,y)\wedge R(y,z))\rightarrow R(x,z)$.
\end{itemize}
\medskip


Let $\sigma'\subseteq \sigma$. The \emph{$\sigma'$-reduct} 
$\mathbf M^{\sigma'}$ of a $\sigma$-structure $\mathbf M$ with domain $M$ is the $\sigma'$-structure with domain $M$, such that
$R(\mathbf M^{\sigma'})=R(\mathbf M)$ for every $R\in\sigma'$.
If $\mathbf N$ is a $\sigma'$-reduct of a $\sigma$-structure $\mathbf M$, then $\mathbf M$ is a \emph{$\sigma$-expansion} of $\mathbf N$ or an \emph{expansion} of $\mathbf N$ by relations in $\sigma\setminus\sigma'$. When $\sigma'=\{R\}$ we use the notation $\mathbf M^R$ instead of $\mathbf M^{\{R\}}$ and we say
that $\mathbf M^R$ is the $R$-reduct of $\mathbf M$.
We extend the above notation to classes of structures: if $\mathscr C$ is a class of structures, then
$\mathscr C^R=\{\mathbf M^R\colon \mathbf M\in\mathscr C\}$.

\subsubsection{Transductions and monadic dependence}

For $k\in\mathbb N$, the \emph{copy operation} $\mathsf C_k$ maps a $\sigma$-structure $\mathbf M$ to the $\sigma\cup\{Z\}$-structure~$\mathsf C_k(\mathbf M)$ obtained by taking $k$ disjoint copies of $\mathbf M$ and making the clones of an element of $\mathbf M$ adjacent in $Z$. 

A unary predicate is also called a \emph{color}. 
A \emph{monadic expansion} or \emph{coloring} of a \mbox{$\sigma$-structure~$\mathbf  M$} is a \mbox{$\sigma^+$-structure~$\mathbf  M^+$}, where $\sigma^+$
is obtained from $\sigma$ by adding unary relations, such 
that $\mathbf  M$ is the $\sigma$-reduct of $\mathbf  M^+$.
For a set $\Sigma$ of colors, the \emph{coloring operation} $\Gamma_\Sigma$ maps a $\sigma$-structure $\mathbf  M$ to the set~$\Gamma_\Sigma(\mathbf  M)$ of all its $\Sigma$-colorings.

A \emph{simple interpretation} $\mathsf I$ of $\tau$-structures in $\sigma$-structures consists  of a $\sigma$-formula $\nu(x)$ and a $\sigma$-formula $\rho_R(\bar x)$ for each  $R\in \tau$ (with arity~$|\bar x|$). For a $\sigma$-structure $\mathbf M$, 
$\mathsf I(\mathbf M)$ is the $\tau$-structure $\mathbf N$ with domain $\nu(\mathbf M)$, such that $R(\mathbf N)=\{\bar a\subseteq\mathbf N\colon \mathbf M\models\rho_R(\bar a)\}$ for every $R\in\tau$.

A \emph{transduction} $\mathsf T$ is a composition
of copy operations, monadic expansions, and simple interpretations.
Every transduction $\mathsf T$ is equivalent to the composition
$\mathsf I\mathbin{\circ}\Gamma_\Sigma\mathbin{\circ} \mathsf C_k$ of a copy operation $\mathsf C_k$, a
coloring operation~$\Gamma_\Sigma$, and a simple interpretation~$\mathsf I$ of $\tau$-structures in $\Sigma$-colored \mbox{$\sigma$-structures} \cite{SBE_TOCL}. 
Hence, for every $\sigma$-structure we have
$\mathsf T(\mathbf  M)=\{\mathsf I(\mathbf  M^+): \mathbf  M^+\in\Gamma_\Sigma(\mathsf C_k(\mathbf  M))\}$. 

For a class $\mathscr M$ of $\sigma$-structures we define $\mathsf T(\mathscr M)=\bigcup_{\mathbf  M\in \mathscr M} \mathsf T(\mathbf  M)$. 
We say that a class $\mathscr N$ can be transduced in a class~$\mathscr M$ if there exists a transduction $\mathsf T$ such that $\mathscr N\subseteq \mathsf T(\mathscr M)$. 
A class $\mathscr C$ is  \emph{monadically dependent} if all its monadic expansions are dependent. 
As shown by Baldwin and Shelah~\cite{BS1985monadic},  a class $\mathscr M$ is monadically dependent if and only if one cannot transduce the class of all finite graphs from $\mathscr M$. 
As a  consequence of the simple observation that transductions compose, we get that all the transductions of a monadically dependent class are monadically dependent. 

\subsection{Graphs and orders}
A \emph{loop} of a binary relation $R$ in a structure $\mathbf M$ is an element $v$ such that $\mathbf M\models R(v,v)$.
As loops do not change properties like monadic dependence and can be encoded as unary relations, we will assume that our binary relations have no loops, that is are anti-reflexive.

By \emph{graph} we mean a relational structure with a single binary symmetric relation (called \emph{adjacency relation}), and
by \emph{oriented graph} we mean a relational structure with a single binary antisymmetric relation (called \emph{arc relation}). An \emph{orientation} of a graph $G$ is an \emph{oriented graph} with Gaifman graph $G$.
A \emph{tournament} is an oriented graph whose arc relation is total. A \emph{transitive tournament} (or \emph{linear order}) is a tournament, whose arc relation is transitive. The binary relation of a linear order will usually be denoted by $<$.
A \emph{local tournament} is an oriented graph such that the in-neighborhood and the out-neighborhood of each vertex induce tournaments.

Let $\mathbf M$ be a relational structure and let $X$ be a subset of the domain of $\mathbf M$. The substructure of $\mathbf M$ \emph{induced by} $X$ is denoted by $\mathbf M[X]$.

In this paper, by an independent set or a clique of a binary structure, we mean an independent set or a clique of its Gaifman graph. Precisely,
a subset $I$ of the domain of $\mathbf M$ is \emph{independent} if there exist no $u\neq v$ in $I$ with $\mathbf M\models R(u,v)$ for some binary relation $R$ of $\mathbf M$.
The \emph{independence number} $\alpha(\mathbf M)$ of $\mathbf M$ is the maximum cardinality of an independent set of $\mathbf M$.
To the opposite, a subset $K$ of the domain of $\mathbf M$ is a \emph{clique} if for every $u\neq v$ in $K$ there is some binary relation $R$ of $\mathbf M$ 
with $\mathbf M\models R(u,v)$.
The \emph{clique number} $\omega(\mathbf M)$ of $\mathbf M$ is the maximum cardinality of a clique of $\mathbf M$.

A \emph{partial order} is an (anti-reflexive) antisymmetric and transitive relation and a \emph{poset} is a relational structure with a single relation that is a partial order.
The \emph{comparability graph} $G(\mathbf P)$ of a poset $\mathbf P$ is the undirected graph with same domain as $\mathbf P$, where two vertices are adjacent if they are distinct and comparable in $\mathbf P$. Hence, the comparability graph of $\mathbf P$ is the Gaifman graph of $\mathbf P$.
A \emph{chain} (resp.\ an \emph{antichain}) of $\mathbf P$ is a set of elements that are pairwise comparable (resp.\ non-comparable). The \emph{height} (resp.\ the \emph{width}) of a poset $\mathbf P$ is the maximum cardinality of a chain (resp.\ an antichain) of $\mathbf P$.
Note, according to our definitions of the clique number and the independence number of a structure,  the height of a poset $\mathbf P$ is $\omega(\mathbf P)$ and its width is $\alpha(\mathbf P)$.

We recall (or introduce) some notations for graphs of interest.
\begin{itemize}
    \item $K_a$ is the \emph{complete graph} of order $a$;
    \item $K_{a,b}$ is the \emph{complete bipartite graph} with parts of size $a$ and $b$;
    \item $S_{a,b}$ the \emph{complete split graph} $\overline{aK_1\cup K_b}$;
    \item $M_a$ is the \emph{matching} $aK_2$;
    \item $L(K_{n,n})$ the line graph of the complete bipartite graph $K_{n,n}$, that is the graph with vertex set $[n]\times[n]$ where $(i,j)$ is adjacent to $(i',j')$ if $i=i'$ or $j=j'$;
    \item $K_k\ctimes K_1$ is the graph obtained by attaching a pendant vertex at each vertex of a clique of size $k$ (the symbol $\ctimes$ denotes the so-called \emph{corona product})
\end{itemize}


We further denote by  $\mathscr B_k$ the class of all the $\{K_{k,k},S_{k,k}\}$-free graphs and by $\obk$ the class of all the orientations of graphs in $\mathscr B_k$.

\newpage
\subsection{Order types}

Let $d\in\mathbb N$ and $\bar a,\bar b\in\mathbb Z^d$.
The \emph{order type} of $(\bar a,\bar b)$ is the vector
$$\otp(\bar a,\bar b)\is (\sgn(b_i-a_i))_{1\leq i\leq d},$$
where $\sgn(x)$ is the sign function defined by
\[
\sgn(x)\is\begin{cases}
	-1&\text{if }x<0,\\
	0&\text{if }x=0,\\
	1&\text{if }x>0.
\end{cases}
\]

We define the set $\mathscr O=\{-1,0,1\}\times\{-1,0,1\}$ of all order types and its subset $\mathscr O^+=\{-1,1\}\times\{-1,1\}$.

The next easy lemma will be useful for working with order types.

\begin{lemma}
	\label{lem:diag}
	For every $\tau_1,\tau_2\in\mathscr O$, every $\tau\in\mathscr O^+$, and  every integer $n$
	there exist $\bar a_1,\dots,\bar a_n,\bar b_1,\dots,\bar b_n\in [2n]\times[2n]$ such that
	\begin{align*}
		\otp(\bar a_i,\bar a_j)&=\tau_1&(1\leq i<j\leq n);\\
		\otp(\bar a_i,\bar b_j)&=\tau&(1\leq i,j\leq n);\\
		\otp(\bar b_i,\bar b_j)&=\tau_2&(1\leq i<j\leq n).
	\end{align*}
\end{lemma}
\begin{proof}
	Let $\bar c_i=(i-1)\tau_1$ and $\bar d_j=2n\tau+(j-1)\tau_2$.
	Then, $\otp(\bar c_i,\bar c_j)=\tau_1$ (for all $1\leq i<j\leq n$),  $\otp(\bar c_i,\bar d_j)=\tau$ (for all $1\leq i,j\leq n$), and $\otp(\bar d_i,\bar d_j)=\tau_2$ (for all $1\leq i<j\leq n$).
	
	Let $X$ be the set of all the first coordinates appearing in the vectors $\bar c_1,\dots,\bar c_n,$ $\bar d_1,\dots,\bar d_n$ and $Y$ be the set of all the second coordinates appearing in the vectors $\bar c_1,\dots,\bar c_n,$ $\bar d_1,\dots,\bar d_n$.
	Let $f_X:X\rightarrow [2n]$ (resp. $f_Y:Y\to[2n]$) be an increasing injection from $X$ (resp. $Y$)
	to $[2n]$.
	Let $\bar a_i$ (resp. $\bar b_j$) be obtained by 
	replacing the first coordinate of $\bar c_i$ (resp. $\bar d_j$) by its image under $f_X$ and the
	second coordinate by its image under $f_Y$.
	Then, it is clear that the vectors $\bar a_i$ and $\bar b_j$ satisfy the same order type properties as the vectors $\bar c_i$ and $\bar d_j$.   
\end{proof}

\subsection{Regular maps}
Let $I$ and $J$ be subsets of $\mathbb Z$ and
let $f$ and $g$ be mappings from $I\times J$ to $M$. We say that 
$(f,g)$ is \emph{regular} in $\mathbf M$ if there exists a mapping $\zeta_{f,g}$, such that in $\mathbf M$ we have
\[
\atp(f(\bar a),g(\bar b))=\zeta_{f,g}(\otp(\bar a,\bar b)).
\]

We say that $f$ is \emph{regular} in $\mathbf M$ if $(f,f)$ is regular in $\mathbf M$. We denote $\dom f$ and $\im f$ the domain and the image of $f$.

Let $(f,g)$ be a regular pair of maps in a $\sigma$-structure $\mathbf M$ and let $R\in\sigma\cup\{=\}$. We define
\begin{align*}
	\trans{f}{g}&\is\{\otp(\bar a,\bar b)\colon \mathbf M\models R(f(\bar a),g(\bar b))\};\\
	\strans{f}{g}&\is\trans{f}{g}\cup\trans{g}{f}.
\end{align*}

\begin{fact}
\label{fact:triv}
As equality is an equivalence relation, we deduce from \zcref{lem:diag} that if $f$ is  regular, then
\begin{itemize}
	\item either $\trans[=]ff\cap\mathscr O^+\neq\emptyset$ and then  
    $\trans[=]ff=\mathscr O$ (i.e. $f$ is a constant map) and we say that $f$ is \emph{constant},
	\item or $\trans[=]ff=\{(0,-1),(0,0),(0,1)\}$ (i.e. $f(\bar a)=f(\bar b)$ if and only if $\bar a$ and $\bar b$ have the same first coordinate), in which case we call $f$ a \emph{repeated line},
	\item or $\trans[=]ff=\{(-1,0),(0,0),(1,0)\}$ (i.e. $f(\bar a)=f(\bar b)$ if and only if $\bar a$ and $\bar b$ have the same second coordinate), in which case we call $f$ a \emph{repeated column},
	\item or $\trans[=]ff=\{(0,0)\}$ (i.e. $f$ is injective), in which case we call $f$ a \emph{regular mesh}.
\end{itemize}
\end{fact}

Let us emphasize some further basic properties of regular maps:
\begin{itemize}
	\item if  $(f,g)$ is regular, then $R$ is antisymmetric if and only if 
	\[\forall \tau\neq(0,0)\qquad\tau\in\trans{f}{g}\Longrightarrow -\tau\notin\trans{g}{f};\]
	\item if $f$ is regular, then $R$ is reflexive if and only if $(0,0)\in\trans{f}{f}$;
	\item if  $f$ is regular, then $R$ is transitive if and only if 
	\[\tau,\tau'\in\trans{f}{f}\Longrightarrow\otp(\tau+\tau')\in\trans{f}{f}.\]
\end{itemize}

We say that $f$ \emph{defines} $\mathbf N$ in $\mathbf M$ if $\mathbf M[\im f]$ is isomorphic to $\mathbf N$.

\subsection{Meshes and transformers} 
In this section, we recall the definitions of meshes and transformers introduced in~\cite{dreier2024flipbreakability}, restating them using the terminology introduced above (order types, regular functions). Then, we introduce the notion of a \emph{full transformer}, which  compensates for minor weaknesses in the definition of transformers.

An \emph{$I\times J$-mesh} of a $\sigma$-structure $\mathbf M$
is an injective mapping $\mu:I\times J\rightarrow M$.
(Hence, a \emph{regular mesh} is an injective regular map.)

For $I\subseteq \mathbb Z$, we define 
\[\mathring{I}=I\setminus\{\min I,\max I\}\]
and we denote by $\mathring{\mu}$ the restriction of $\mu$ to
$\mathring{I}\times\mathring{J}$.

Recall that a pair mesh $(\mu,\mu')$ of $I\times J$-meshes is \emph{regular}
if there exists a map $F$ such that
$$\atp(\mu(\bar a),\mu'(\bar b))=F(\otp(\bar a,\bar b)).$$
In the case where $F$ is a constant map, the pair $(\mu,\mu')$ is \emph{homogeneous}. A pair of $I\times J$-meshes is \emph{conducting} if it is regular but not homogeneous.

A regular $I\times J$-mesh $\mu$ is \emph{vertical} if
there exists a function $\alpha:I\rightarrow M$ and a non-constant map $F$ such that
\[
\atp(\alpha(i),\mu(i',j'))=F(\otp(i,i')).
\]
The function $\alpha$ is a \emph{vertical guard} of $\mu$.

Equivalently, $\mu$ is vertical if there exists a repeated line $\mu':I\times J\rightarrow M$ such that $(\mu',\mu)$ is conducting. This repeated line can be defined by $\mu'(i,j)\is\alpha(i)$.

A regular $I\times J$-mesh $\mu$ is \emph{horizontal} if
there exists a function $\beta:J\rightarrow M$ and a non-constant map $F$ such that
\[
\atp(\mu(i,j),\beta(j'))=F(\otp(j,j')).
\]
The function $\beta$ is a \emph{horizontal guard} of $\mu$.

Equivalently, $\mu$ is horizontal if there exists a repeated column $\mu':I\times J\rightarrow M$ such that $(\mu,\mu')$ is conducting. This repeated column can be defined by $\mu'(i,j)\is\beta(j)$.

A \emph{minimal transformer} of \emph{length} $h$ and \emph{order} $n$ of a $\sigma$-structure $\mathbf M$
is a sequence $(\mu_1,\dots,\mu_h)$ of $[n]\times[n]$-meshes such that 
\begin{enumerate}
	\item $\mu_s$ is vertical if and only if $s=1$;
	\item $\mu_s$ is horizontal if and only if $s=h$;
	\item $\mu_s\neq\mu_t$ if $s\neq t$;
	\item for all $1\leq s,t\leq h$, $(\mu_s,\mu_t)$ is regular if $s=t$, conducting if $|s-t|=1$, and homogeneous otherwise.
\end{enumerate}

The importance of minimal transformers is clear from the following characterization of monadically dependent classes of binary structures.

\begin{etheorem}[\cite{dreier2024flipbreakability}]
	\label{thm:trsf}
	A class $\mathscr C$ of binary structures is monadically dependent if and only if for every integer $h\in\mathbb N$ there is some integer $n$, such that no structure $\mathbf M\in\mathscr C$ contains a minimal transformer of length $h$ and order $n$.
\end{etheorem}

We now introduce the notion of a \emph{full transformer}. A minor weakness in the definition of transformer is that the set of vertices witnessing that the first (resp.\ the last) mesh is vertical (resp.\ horizontal) is not part of the transformer and lacks some regularity properties. 

\begin{definition}
A \emph{full transformer} of length $h$ and order $n$ of $\mathbf M$ is a sequence
$(\mu_0,\dots,\mu_{h+1})$, where $(\mu_1,\dots,\mu_h)$ is a minimal transformer, $\mu_0$ is a repeated line, $\mu_{h+1}$ is a repeated column, and  $(\mu_0,\mu_1)$ and $(\mu_h,\mu_{h+1})$ are both conducting. 
A full transformer of length $h$ and order $n$ of $\mathbf M$ is \emph{tight} if $\mathbf M$ has no minimal transformer of length $h'<h$ and order at least $n-2$.
\end{definition}

The notion of tightness is stronger that the original notion of minimality of transformers. The
motivation for the introduction of this stronger property is that sometimes it is possible to restrict a non-vertical/horizontal mesh to a vertical/horizontal one by removing at most two coordinates, which is not such a big change when we consider transformers of very large orders.

The next lemma is an immediate corollary of \zcref{thm:trsf}.

\begin{lemma}
\label{lem:ftrsf}
	A class $\mathscr C$ of binary structures is monadically independent if and only if there exists an  integer $h\in\mathbb N$ and an integer $n_0$,  such that for every integer $n\ge n_0$, some structure $\mathbf M\in\mathscr C$ contains a tight full transformer of length $h$ and order $n$.
\end{lemma}
\begin{proof}
	Assume $\mathscr C$  is monadically independent. According to \zcref{thm:trsf}, there exists some integer $h$ such that for every integer $n$ the class $\mathscr C$ contains a minimal transformer of length $h$ and order $n$. Choose $h$ minimal with this property
    and let $n_0$ be minimum such that no transformer of length $h'<h$ and order at least $n_0-2$ exists  in some structure in $\mathscr C$.
	Then, the full transformers defined from minimal transformers of length $h$ and order $n\ge n_0$ is automatically tight.
	
The converse implication is a straightforward consequence of \zcref{thm:trsf}, as every full transformer defines a minimal transformer.
\end{proof}

\section{Moving from one \texorpdfstring{$\{K_{k,k},S_{k,k}\}$}{\ifpdfstringunicode{\{K\unichar{"2096}\unichar{"2096},S\unichar{"2096}\unichar{"2096}\}}{\{Kkk,Skk\}}}-free relation to another}
\label{sec:main}
\subsection{Full transformers in oriented \texorpdfstring{$\{K_{k,k},S_{k,k}\}$}{\ifpdfstringunicode{\{K\unichar{"2096}\unichar{"2096},S\unichar{"2096}\unichar{"2096}\}}{\{Kkk,Skk\}}}-free graphs}

Transformers can be quite complex even in binary relational structures. This complexity can be strongly reduced by forbidding some local patterns.
The exclusion of complete bipartite graphs as subgraphs of the Gaifman graphs would reduce the study to nowhere dense structures. Here we shall forbid a bit less, allowing arbitrarily large cliques, by excluding only  an induced biclique and an induced complete split graph.

\begin{lemma}
\label{lem:2tour}
Let $(f,g)$ be a regular pair of non-constant regular maps from $[n]\times [n]$ to $M$
and let $R\in\sigma$ be an antisymmetric relation of $\mathbf M$.
Assume that $\mathbf M^R\in\obk$,  where $n\geq 2k$, and that either $f=g$ or $\im f\cap\im g=\emptyset$.

If  $\strans{f}{g}\cap \mathscr O^+ \neq \emptyset$,
then $f$ and $g$ define tournaments in  $\mathbf M^R$.
\end{lemma}
\begin{proof}
Let $E(x,y)\is R(x,y)\vee R(y,x)$
and   $\tau\in\strans{f}{g}\cap \mathscr O^+$.  
As $f$ is non-constant, there exists $\tau_1\in\mathscr O\setminus \trans[=]ff$. 
Assume for contradiction that $g$ is not a tournament. 
Then, there exists $\tau_2\in\mathscr O$ such that neither $\tau_2$ nor $-\tau_2$ belongs to $\trans{g}{g}\cup \trans[=]gg$.

According to \zcref{lem:diag} there exist $\bar a_1,\dots,\bar a_k,\bar b_1,\dots,\bar b_k\in [2k]\times[2k]$ such that $\otp(\bar a_i,\bar a_j)=\tau_1$ for all $1\leq i<j\leq k$, $\otp(\bar a_i,\bar b_j)=\tau$ for all $i,j\in[k]$, and $\otp(\bar b_i,\bar b_j)=\tau_2$ for all $1\leq i<j\leq k$.
(Note that  $f(\bar a_1),\dots,f(\bar a_k),g(\bar b_1),\dots,g(\bar b_k)$ are all distinct since $\tau_1\notin\strans[=]{f}{f},\tau\notin\strans[=]{f}{g}$ when $f=g$ according to \zcref{fact:triv}, and $\tau_2\notin\strans[=]{g}{g}$.)
As $\tau\in\strans{f}{g}$, we have $E(f(\bar a_i),g(\bar b_j))$ for all $1\leq i,j\leq 2k$.
As neither $\tau_2$ nor $-\tau_2$ belongs to $\trans{g}{g}$, we have
$\neg E(g(\bar b_i),g(\bar b_j))$ for all $1\leq i<j\leq 2k$.
In particular, as $f(\bar a_1),\dots,f(\bar a_k)$ either form an independent set or induce a transitive tournament in $\mathbf M^R$, 
and as $g(\bar b_1),\dots,g(\bar b_k)$ are independent in $\mathbf M^R$, 
$f(\bar a_1),\dots,f(\bar a_k),g(\bar b_1),\dots,g(\bar b_k)$  induce a $K_{k,k}$ or
an $S_{k,k}$ in $\Gaif(\mathbf M^R)$, contradicting our assumption. 

Thus, $g$ defines a tournament in $\mathbf M^R$ and, by symmetry, so does $f$.
\end{proof}

\begin{lemma}
	\label{lem:tour}
	Let $f:[n]\times[n]\rightarrow M$ be a regular mesh.
	If $f$ defines a tournament in  $\mathbf M^R$, then 
	$\mathring{f}$ is vertical or horizontal in $\mathbf M^R$, and it is both horizontal and vertical in $\mathbf M^R$ if the tournament is not transitive.
\end{lemma}
\begin{proof}
	By reversing the order of the index sets and exchanging the coordinates if necessary, we can assume that $\trans{f}{f}$
	contains $(0,1),(1,0)$, and $(1,-1)$.
	
	Assume $f$ defines a transitive tournament in $\mathbf M^R$.
	Then, the transitivity of $R$ in $\im f$ implies 
	\[\trans{f}{f}\setminus\{(0,0)\}=\{(1,0),(1,-1),(1,1),(0,1)\}.\]
	In other words, if $(i,j)\neq (i',j')$, then 
	$\mathbf M\models R(f(i,j),f(i',j'))$ if and only if $(i,j)$ precedes $(i',j')$ in the lexicographic order on $[n]\times[n]$.
	Thus, $R(f(i,1),f(i',j'))\leftrightarrow (i\leq i')$ 
	(and $R(f(i',j'),f(i,1))\leftrightarrow (i> i')$)
	whenever $j'>1$. Hence, $\mathring{f}$ is vertical. 
	
	Assume $f$ defines a non-transitive tournament in $\mathbf M^R$.
	Then, $(-1,-1)\in\trans{f}{f}$. 
	Then, for $i,j,i',j'>1$, $R(f(i,1),f(i',j'))$ holds exactly if $i=i'$ and
	$R(f(1,j),f(i',j'))$ holds exactly if $j\leq j'$, hence $\mathring{f}$ is both horizontal and vertical. 
\end{proof}

\begin{lemma}
	\label{lem:HV}
	Let $(f,g)$ be a regular pair of regular $[n]\times[n]$-meshes of a $\sigma$-structure $\mathbf M$ and let $R\in\sigma$ be an antisymmetric relation in $\mathbf M$.
	
	If $\strans{f}{g}\cap \mathscr O^+=\emptyset$ and $(0,1)\in \strans{f}{g}$,
	then $\mathring{f}$ and $\mathring{g}$  are both vertical in $\mathbf M^R$.
\end{lemma}
\begin{proof}
	By exchanging $f$ and $g$ and reversing the order of an index set  if necessary, we can assume that $(0,1)\in\trans{f}{g}$ (hence $(0,-1)\notin\trans{g}{f}$).
	Then, for every $j'>1$ we have
	\[
	\neg R(g(i',j'),f(i,1))\qquad\text{and}\qquad R(f(i,1),g(i',j'))\leftrightarrow (i=i'),
	\]
	and for every $j'<n$ we have
	\[
	\neg R(g(i,n),f(i',j'))\qquad\text{and}\qquad R(f(i',j'),g(i,n))\leftrightarrow (i=i').
	\]
	Hence, $\mathring{f}$ and $\mathring{g}$  are both vertical in $\mathbf M^R$.
\end{proof}

\begin{lemma}
\label{lem:IHVK}
Let $f:[n]\times[n]\rightarrow M$ be a regular mesh
and let $R\in\sigma$ be an antisymmetric relation of $\mathbf M$.
Assume that $\mathbf M^R\in\obk$,  where $n\geq 2k$,
    
	
	Then, $\Gaif(\mathbf M^R)[\im f]$ is 
	either $n^2K_1$ (i.e. $\im f$ is  independent in $\mathbf M^R$),  $K_{n^2}$, $nK_n$,  or $L(K_{n,n})$. In the three last cases, $\omega(\Gaif(\mathbf M^R)[\im f])\geq n$ and  $\mathring{f}$ is horizontal or vertical in $\mathbf M^R$.
\end{lemma}

\begin{proof}
	Assume $\Gaif(\mathbf M^R)[\im f]$ is  not  an independent set.
	If $\trans{f}{f}\cap\mathscr O^+\neq\emptyset$, then 
    the result follows from \zcref{lem:2tour,lem:tour}. Otherwise, if $\trans{f}{f}$ contains $(0,1)$ or $(0,-1)$ (resp. $(1,0)$ or $(-1,0)$) , then 
    $\mathring{f}$ is vertical (resp. horizontal) according to \zcref{lem:HV} (after possible reversal of the index set and/or transposition of $f$). The structure of the Gaifman graph  $\Gaif(\mathbf M^R)[\im f]$ is clear  when $\trans{f}{f}\cap\mathscr O^+=\emptyset$: it is either $nK_n$ if $\trans{f}{f}$ does not contain both an order type in  $\{(0,1),(0,-1)\}$ and an order type in  $\{(1,0),(-1,0)\}$, and $L(K_{n,n})$, otherwise.
\end{proof}

\begin{lemma}
\label{lem:Hstar}
    Let $(\mu_0,\dots,\mu_{h+1})$ be a  full transformer of $\mathbf M$ with order $n\geq 2k$. 
    Assume $R$ is antisymmetric on $\mathbf M$ and $\mathbf M^R\in\obk$.
    \begin{itemize}
        \item If $\im\mu_1$ is independent in $\mathbf M^R$, then 
        \[\text{either}\quad\strans{\mu_0}{\mu_1}=\emptyset\quad\text{or}\quad\strans{\mu_0}{\mu_1}=\{0\}\times\{-1,0,1\};\] 
        \item if $\im\mu_h$ is independent in $\mathbf M^R$, then 
        \[\text{either}\quad\strans{\mu_h}{\mu_{h+1}}=\emptyset\quad\text{or}\quad\strans{\mu_h}{\mu_{h+1}}=\{-1,0,1\}\times\{0\}.\]
    \end{itemize}
    \end{lemma}
\begin{proof}
    As $\mu_0$ is a repeated line and $(\mu_0,\mu_1)$ is regular, 
    $\strans{\mu_0}{\mu_1}=T\times \{-1,0,1\}$ for some $T\subseteq\{-1,0,1\}$.
    As $\im\mu_1$ is independent in $\mathbf M^R$, according to \zcref{lem:2tour}, $\strans{\mu_0}{\mu_1}\cap\mathscr O^+=\emptyset$.
    Thus, $\strans{\mu_0}{\mu_1}\subseteq\{0\}\times \{-1,0,1\}$.

    The second item is deduced by considering the transpose of the full transformer.
\end{proof}

\subsection{From one relation to another}

 In this section, we consider $\sigma$-structures $\mathbf M$ such that 
$\mathbf M^R$ and $\mathbf M^{R'}$
 belong to $\obk$. This section is devoted to the proof of the technical \zcref{lem:main}, which allows to deduce the monadic dependence of an expansion of a monadically dependent class of binary structures, when the introduced relation $R$  is not too far from an existing relation $R'$ with Gaifman graph in $\mathscr B_k$. The ``not too far'' condition will be formalized using the next definition.

\begin{definition}[See \zcref{fig:pi}]
Let $R,R'\in\sigma$. A $\sigma$-structure
    $\mathbf M$ has the $\Pi^k_{R,R'}$  property if, for every subset $X$ of vertices of $\mathbf M$, 
\begin{itemize}
    \item if $\Gaif(\mathbf M^R)[X]$ is $K_k$, then 
    $\Gaif(\mathbf M^{R'})[X]$ is not $k\,K_1$;
    \item if $\Gaif(\mathbf M^R)[X]$ is $K_k\ctimes K_1$ or $M_k$, 
    then 
    $\Gaif(\mathbf M^{R'})[X]$ is neither $K_k\cup k\,K_1$ nor $2k\,K_1$.
\end{itemize}    
\end{definition}

\begin{figure}[ht]
    \centering
    \includegraphics[width=\linewidth]{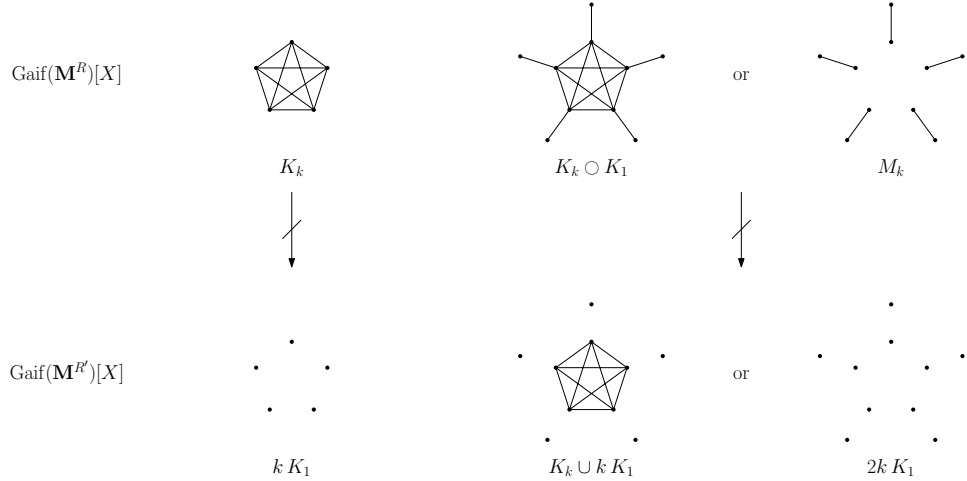}
    \caption{The forbidden modifications from $\mathbf M^R$ to $\mathbf M^{R'}$.}
    \label{fig:pi}
\end{figure}

\begin{lemma}
\label{lem:main}
	Let $R,R'\in\sigma$. Assume  
	that $R$ and $R'$ are both antisymmetric in $\mathbf M$, that 
	$\mathbf M^R$ and $\mathbf M^{R'}$ belong to $\obk$, and that $\mathbf M$ has the property $\Pi^k_{R,R'}$.
	
	If $\mathbf M$ contains a tight full transformer $(\mu_0,\dots,\mu_{h+1})$ of order $n\geq 2k$, then
	$\mathbf M^{\sigma\setminus\{R\}}$ or $\mathbf M^{\{R,R'\}}$ contains a transformer of length at most $h$ and order at least $n-2$.
\end{lemma}
\begin{proof}
Let $(\mu_1,\dots,\mu_h)$ be the minimal transformer associated to the tight full transformer.

\begin{claim}
\label{cl:1}
    Let $s\in [h]$. If $\mu_s$ is not independent in $\mathbf M^{R}$ then $s\in\{1,h\}$ and $\mu_s$ is not independent in $\mathbf M^{R'}$, thus $\mathring{\mu}_{s}$ is horizontal or vertical in $\mathbf M^{R'}$.
\end{claim}
\begin{clproof}
    According to \zcref{lem:IHVK}, $\mathring{\mu}_{s}$ is horizontal or vertical and the graph $\Gaif(\mathbf M)[\im\mu_s]$ contains a clique of size $n$. By tightness, $s\in\{1,h\}$.
    As 
    $\mathbf M$ has the property $\Pi^k_{R,R'}$, $\Gaif(\mathbf M^{R'})[\im\mu_s]$ is not independent.
    According to \zcref{lem:IHVK}, we deduce that $\mathring{\mu}_{s}$ is horizontal or vertical in $\mathbf M^{R'}$.
\end{clproof}

\begin{claim}
\label{cl:3}
    If $\mu_1$ is independent and vertical in $\mathbf M^R$, then 
    $\mathring{\mu}_{1}$ is horizontal or vertical in $\mathbf M^{R'}$.
\end{claim}
\begin{clproof}
  If $\mu_1$ is not independent in $\mathbf M^{R'}$ the result follows from \zcref{lem:IHVK}.
    Otherwise,
  as $(\mu_0,\mu_1)$ is conducting in $\mathbf M^R$,
    according to \zcref{lem:Hstar} we have 
    $\strans{\mu_0}{\mu_1}=\{0\}\times\{-1,0,1\}$.
    In particular, $\Gaif(\mathbf M^{R})[\im \mu_0\cup\im\mu_1]$ contains either $M_n$ or $K_n\ctimes K_1$ as an induced subgraph.
    If $(\mu_0,\mu_1)$ is not conducting in $\mathbf M^{R'}$ then, 
    according to \zcref{lem:Hstar}, 
    $\strans[R']{\mu_0}{\mu_1}=\emptyset$, contradicting the assumption that $\mathbf M$ has the property $\Pi^k_{R,R'}$. Hence, $(\mu_0,\mu_1)$ is  conducting in $\mathbf M^{R'}$, thus 
    $\mu_1$ (hence $\mathring{\mu}_{1}$) is vertical in $\mathbf M^{R'}$.
\end{clproof}
\begin{claim}
\label{cl:4}
    Assume $h\geq 2$, $s\in[h-1]$, and at least one of $\mu_s$ and $\mu_{s+1}$ is independent in $\mathbf M^R$.
    If $(\mu_s,\mu_{s+1})$ is conducting in $\mathbf M^R$, then either $(\mu_s,\mu_{s+1})$ is conducting in $\mathbf M^{R'}$ or $h=2$ and $(\mathring{\mu}_s,\mathring{\mu}_{s+1})$ is a minimal transformer in $\mathbf M^{\{R,R'\}}$.
\end{claim}
\begin{clproof}
According to \zcref{lem:2tour}, $\strans{\mu_s}{\mu_{s+1}}\cap\mathscr O^+=\emptyset$. 
If $(0,1)\in\strans{\mu_s}{\mu_{s+1}}$, then by \zcref{lem:HV},
$\mathring{\mu}_{s+1}$ is vertical in $\mathbf M^R$, contradicting the tightness of the full transformer.
It follows that $\strans{\mu_s}{\mu_{s+1}}=\{(0,0)\}$ (as $(\mu_s,\mu_{s+1})$ is conducting in $\mathbf M^R$).

If none of $\mu_s$ and $\mu_{s+1}$ is independent in $\mathbf M^{R'}$, then, according to \zcref{lem:IHVK}, $\mathring{\mu}_s$ and $\mathring{\mu}_{s+1}$ are either horizontal or vertical in
$\mathbf M^{R'}$. Hence, one of $(\mathring \mu_1, \ldots, \mathring \mu_s)$, $(\mathring \mu_s, \mathring \mu_{s+1})$ or $(\mathring \mu_{s+1}, \ldots, \mathring \mu_{h})$ is a transformer in $\mathbf M^{\{R,R'\}}$. 
The first and last case contradict the tightness of $(\mu_0, \ldots, \mu_{h+1})$ in $\mathbf M$, while the second case implies $h=2$ (by tightness).

Otherwise, at least one of $\mu_s$ and $\mu_{s+1}$ is independent in $\mathbf M^{R'}$.
Without loss of generality, we can assume that $\mu_{s+1}$ is independent in $\mathbf M^{R'}$. 
According to property $\Pi^k_{R,R'}$ and \zcref{lem:IHVK}, $\mu_{s+1}$ is also independent in $\mathbf M^R$.

Assume for contradiction that $(\mu_s,\mu_{s+1})$ is not conducting in $\mathbf M^{R'}$. 
Then $(\mu_s,\mu_{s+1})$ is homogeneous  in $\mathbf M^{R'}$, thus $\strans[R']{\mu_s}{\mu_{s+1}}$ is either $\mathscr O$ or $\emptyset$.
As the first possibility is ruled out by \zcref{lem:2tour}, we deduce $\strans[R']{\mu_s}{\mu_{s+1}}=\emptyset$.

If $\mu_s$ is independent in $\mathbf M^{R'}$ then, as above, it is independent in $\mathbf M^R$ and, by property $\Pi^k_{R,R'}$, as $\strans[R]{\mu_s}{\mu_{s+1}} = \{(0,0)\}$, we cannot have $\strans[R']{\mu_s}{\mu_{s+1}}=\emptyset$.

Otherwise, $\mu_s$ is not independent in $\mathbf M^{R'}$. Since $\mu_{s+1}$ is independent in $\mathbf M^{R'}$, according to \zcref{lem:IHVK},
 there exists a line or a column $X$ of $\mu_s$ and the matching line or column $Y$ in $\mu_{s+1}$ such that
$\Gaif(\mathbf M^{R'})[X]$ is isomorphic to $K_n$ and
$\Gaif(\mathbf M^{R'})[X\cup Y]$ is isomorphic to
$K_n\cup n\,K_1$. 
On the other hand, according to \zcref{lem:IHVK}, $X$ is either an independent set or a clique in $\Gaif(\mathbf M^{R})$, thus $\Gaif(\mathbf M^{R})[X\cup Y]$ is  isomorphic to $K_n\ctimes K_1$ or $M_n$, contradicting  property $\Pi^k_{R,R'}$ (as $n\ge k$).
\end{clproof}

We consider three different cases.
\begin{itemize}
    \item {\bf Case 1: $h=1$.} In this case there exist $R_1,R_2\in\sigma$ such that $\mu_1$ is vertical in $\mathbf M^{R_1}$ 
    (precisely, such that $(\mu_0,\mu_1)$ is conducting in $\mathbf M^{R_1}$)
    and horizontal in $\mathbf M^{R_2}$     (precisely, such that $(\mu_h,\mu_{h+1})$ is conducting in $\mathbf M^{R_2}$). If $R\notin\{R_1,R_2\}$, then $(\mu_1)$ is a transformer of $\mathbf M^{\sigma\setminus\{R\}}$ of order $n$. If $R_1=R_2=R$, then 
    $(\mu_1)$ is a transformer in $\mathbf M^{\{R,R'\}}$. By considering the transpose of the transformer if necessary, we can assume $R_1=R$.

    According to either \zcref{cl:1} (if $\mu_1$ is not independent in $\mathbf M^R$) or  \zcref{cl:3} (if $\mu_1$ is  independent in $\mathbf M^R$),  $\mathring{\mu}_{1}$ is horizontal or vertical in $\mathbf M^{R'}$.
    If $\mathring{\mu}_{1}$ is horizontal in $\mathbf M^{R'}$, then
    $(\mathring{\mu}_{1})$ is a transformer of order $n-2$ in $\mathbf M^{\{R,R'\}}$. Otherwise, $\mathring{\mu}_{1}$ is vertical in $\mathbf M^{R'}$ and $(\mathring{\mu}_{1})$ is a  transformer of order $n-2$ in $\mathbf M^{\{R',R_2\}}$, hence in $\mathbf M^{\sigma\setminus\{R\}}$.
    \item {\bf Case 2: $h=2$ and none of $\mu_1$ and $\mu_2$ is independent in $\mathbf M^R$.}
    According to \zcref{lem:IHVK} and the tightness of the transformer, $\mathring{\mu}_{1}$  and 
    $\mathring{\mu}_{2}$ are respectively  vertical and horizontal in $\mathbf M^R$. 
    According to \zcref{cl:1} and the tightness of the full transformer,  $\mathring{\mu}_{1}$ and $\mathring{\mu}_{2}$ are respectively  vertical and horizontal in $\mathbf M^{R'}$.
    Moreover, there exists a relation $R_2\in\sigma$ such that $(\mu_1,\mu_2)$ is conducting in $\mathbf M^{R_2}$. Thus, $(\mathring{\mu}_{1},\mathring{\mu}_{2})$ is a transformer of order $n-2$ both in
    $\mathbf M^{\{R,R_2\}}$ and in $\mathbf M^{\{R',R_2\}}$, hence 
    a transformer of order $n-2$ in
    $\mathbf M^{\{R,R'\}}$ (if $R_2\in\{R,R'\}$) or in $\mathbf M^{\sigma\setminus\{R\}}$ (otherwise).
    \item {\bf Case 3: $h>2$ or $h=2$ and $\mu_1$ or $\mu_2$ is independent in $\mathbf M^R$.}
    If $\mu_1$ is vertical in $\mathbf M^R$, then 
    by \zcref{cl:1,cl:3}, $\mathring{\mu}_{1}$ is horizontal or vertical in $\mathbf M^{R'}$, hence vertical in $\mathbf M^{R'}$ since the full transformer is tight;
    Similarly, if $\mu_h$ is horizontal in $\mathbf M^R$, then $\mathring{\mu}_{h}$ is horizontal in $\mathbf M^{R'}$.
    Moreover, according to \zcref{cl:4}, if a pair $(\mu_s,\mu_{s+1})$ is conducting in $\mathbf M^R$, either $(\mathring{\mu}_s,\mathring{\mu}_{s+1})$ is a transformer in $\mathbf M^{\{R,R'\}}$ or $(\mu_s,\mu_{s+1})$
     is conducting in $\mathbf M^{R'}$. If the first case never occurs,
     it follows that $(\mathring{\mu}_1,\dots,\mathring{\mu}_h)$ is a transformer in $\mathbf M^{\sigma\setminus\{R\}}$.
\end{itemize}
\end{proof}
\begin{theorem}[store=thm:main]
    Let $R,R'\in\sigma$ and let $\mathscr C$ be a class of $\sigma$-structures.
    If $\mathscr C^R\subseteq \obk$, $\mathscr C^{R'}\subseteq \obk$, and every structure in $\mathscr C$ has the $\Pi^k_{R,R'}$ property, then the following are equivalent:
    \begin{enumerate}
        \item\label{en:1} $\mathscr C$ is monadically dependent;
        \item\label{en:2} $\mathscr C^{\sigma\setminus\{R\}}$ and $\mathscr C^{\{R,R'\}}$ are monadically dependent;
    \end{enumerate}
\end{theorem}
\begin{proof}
The implication $\eqref{en:1}\Rightarrow \eqref{en:2}$ 
is obvious.
Hence, the theorem will follow from the proof that  \eqref{en:2} implies \eqref{en:1}.

Assume for contradiction that \eqref{en:2} holds but \eqref{en:1} does not.
According to \zcref{lem:ftrsf} there exists integers $h$ and $n_0$ such that for every $n\ge n_0$ some $\mathbf M\in\mathscr C$ contains a tight full transformer of length $h$ and order $n$.
According to \zcref{lem:main},  $\mathbf M^{\{R,R'\}}$ or $\mathbf M^{\sigma\setminus\{R\}}$ contains a transformer of length $h$ and order at least $n-2$. Thus, it follows from \zcref{thm:trsf} that $\mathscr C^{\{R,R'\}}$ or $\mathscr C^{\sigma\setminus\{R\}}$ is monadically independent, contradicting~\eqref{en:2}.
\end{proof}
\section{Applications}
\label{sec:app}
\subsection{Reorientations of \texorpdfstring{$\{K_{k,k},S_{k,k}\}$}{\ifpdfstringunicode{\{K\unichar{"2096}\unichar{"2096},S\unichar{"2096}\unichar{"2096}\}}{\{Kkk,Skk\}}}-free graphs}
\label{sec:reorient}

In this section, we consider (binary) $\sigma$-structures with a special relation 
$R$ encoding the orientation of a $\{K_{k,k},S_{k,k}\}$-free graph.
Note that when $R'$ is a reorientation of $R$, the Gaifman graph of the $R$-reduct is the same as the one of the $R'$-reduct, hence the property $\Pi^k_{R,R'}$ trivially holds.

\begin{definition}
\label{def:overlay}
For two orientations $\mathbf G_1,\mathbf G_2$ of a graph $G$, we denote by $\mathbf G_1\overlay \mathbf G_2$ the 
\emph{overlay} of $\mathbf G_1$ and $\mathbf G_2$, which is the 
$\{R_1,R_2\}$-structure $\mathbf G$, such that $\mathbf G^{R_1}$ is $\mathbf G_1$ (with arc relation renamed $R_1$) and
$\mathbf G^{R_2}$ is $\mathbf G_2$ (with arc relation renamed $R_2$).
\end{definition}

We denote by $\mathsf G$ a sequence $(\mathbf G_i\colon i<\omega)$ of structures and we write $\mathsf G\subset\mathscr C$ if $\mathbf G_i\in\mathscr C$ for every $i<\omega$.
We say that a sequence $\mathsf G$ is \emph{monadically dependent} if $\{\mathbf G_i\colon i<\omega\}$ is monadically dependent.

\begin{definition}
    We say that a monadically dependent sequence $\mathsf G$ of oriented graphs is a \emph{compatible reorientation} of a monadically dependent sequence $\mathsf H$ of oriented graphs, and we write $\mathsf G\simor\mathsf H$ if each 
    $\mathbf H_i$ is a reorientation of $\mathbf G_i$ and
    $\mathsf G\overlay \mathsf H\is(\mathbf G_i\overlay\mathbf H_i\colon i<\omega)$ is monadically dependent.
\end{definition}

The next theorem witnesses that the  relation $\simor$ provides a sufficient condition ensuring that the reorientation of a relation defining an
oriented $\{K_{k,k},S_{k,k}\}$-free graph preserves monadic dependence.

\begin{theorem}[store=thm:reorient]
Let $R\in\sigma$.
    Let $\mathsf M$ and $\mathsf N$ be sequences of $\sigma$-structures such that $\mathbf M_i^{\sigma\setminus \{R\}}=\mathbf N_i^{\sigma\setminus \{R\}}$, 
    $\mathbf M_i^R\in\obk$ and 
    $\mathbf N_i^R$ is a reorientation of $\mathbf M_i^R$.

If $\mathsf M^R\simor\mathsf N^R$, then
$\mathsf M$ is monadically dependent if and only if $\mathsf N$ is monadically dependent.
\end{theorem}
\begin{proof}
    (For the definition of the overlay, which is an $\{R_1,R_2\}$-structure, refer to \zcref{def:overlay}.)
    Let $\tau=\sigma\setminus\{R\}\cup\{R_1,R_2\}$ and let
    $\mathbf S_i$ be the $\tau$-structure defined by 
    $\mathbf S_i^{\{R_1,R_2\}}=\mathbf M_i^R\overlay\mathbf N_i^R$ and 
    $\mathbf S_i^{\tau\setminus\{R_1,R_2\}}=\mathbf M_i^{\sigma\setminus\{R\}}$.
    As $\mathsf M^R\simor \mathsf N^R$, $\{\mathbf S_i^{\{R_1,R_2\}}\colon i<\omega\}$ is monadically dependent.
    Thus, according to 
    \zcref{thm:main}, if $\mathsf M$ is monadically dependent, then 
    $\mathsf S$ is monadically dependent, thus so is its reduct $\mathsf N$.
    The theorem then follows from the obvious observation that $\simor$ is a symmetric relation.
\end{proof}

\begin{corollary}
Let $k\in\mathbb N$ be fixed.
The relation $\simor$ is an equivalence relation on sequences of oriented graphs in $\obk$.
\end{corollary}
\begin{proof}
    The relation $\simor$ is obviously reflexive and symmetric. Let us prove that it is also transitive.
    Assume $\mathsf F\simor \mathsf G$ and $\mathsf G\simor \mathsf H$.
    According to \zcref{thm:reorient}, 
    the condition $\mathsf G\simor \mathsf H$ implies that the sequence $\mathsf F\overlay \mathsf G$ is monadically dependent if and only if $\mathsf F\overlay \mathsf H$ is monadically dependent. Hence, $\mathsf F\simor \mathsf H$.
\end{proof}
\subsection{Replacing an oriented graph with small independence number}
\label{sec:replace}
We now consider the case where a relation $R$ defines the orientation of a graph with independence number less than $k$.
As the forbidden configurations defining the property $\Pi^k_{R,R'}$ all involve an independent set of size $k$, if one considers another relation $R'$ defining an orientation of a graph with independence number less than $k$, then the property $\Pi^k_{R,R'}$ trivially holds.

Let $k\in\mathbb N$ be fixed and let $\osk$ be the class of all the orientations of graphs with independence number less than $k$.
We define a relation $\sima$ analogous to the relation $\simor$ introduced in the previous section.

\begin{definition}
    We say that a monadically dependent sequence $\mathsf G\subset\osk$ is a \emph{compatible replacement} of a monadically dependent sequence $\mathsf H\subset\mathscr \osk$, and we write $\mathsf G\sima\mathsf H$ if each 
    $\mathbf H_i$ has the same domain as $\mathbf G_i$ and
    $\mathsf G\overlay \mathsf H\is(\mathbf G_i\overlay\mathbf H_i\colon i<\omega)$ is monadically dependent.
\end{definition}

\begin{theorem}[store=thm:calpha]
Let $R\in\sigma$.
    Let $\mathsf M$ and $\mathsf N$ be sequences of $\sigma$-structures such that $\mathbf M_i^{\sigma\setminus \{R\}}=\mathbf N_i^{\sigma\setminus \{R\}}$  and 
    $\mathbf M_i^R,\mathbf N_i^R\in\osk$.

If $\mathsf M^R\sima\mathsf N^R$, then
$\mathsf M$ is monadically dependent if and only if $\mathsf N$ is monadically dependent.
\end{theorem}
\begin{proof}
    The proof is analogous to the one of \zcref{thm:reorient}.
    
    Let $\tau=\sigma\setminus\{R\}\cup\{R_1,R_2\}$ and let
    $\mathbf S_i$ be the $\tau$-structure defined by 
    $\mathbf S_i^{\{R_1,R_2\}}=\mathbf M_i^R\overlay\mathbf N_i^R$ and 
    $\mathbf S_i^{\tau\setminus\{R_1,R_2\}}=\mathbf M_i^{\sigma\setminus\{R\}}$.
    As $\mathsf M^R\sima \mathsf N^R$, $\{\mathbf S^{\{R_1,R_2\}}\colon i<\omega\}$ is monadically dependent.
    Thus, according to 
    \zcref{thm:main}, if $\mathsf M$ is monadically dependent, then 
    $\mathsf S$ is monadically dependent, thus so is its reduct $\mathsf N$.
    The reverse implication follows from the obvious observation that $\sima$ is a symmetric relation.
\end{proof}
\begin{corollary}
    Let $k$ be fixed.
    The relation $\sima$ is an equivalence relation on sequences of oriented graphs belonging to $\osk$.
\end{corollary}

\subsection{Characterization of classes with bounded twin-width}
We now extend \zcref{thm:tww4} by considering expansion by an oriented graph with bounded independence number instead of an expansion by a transitive tournament (i.e., linear order).





\begin{theorem}[store=thm:nalpha]
 A class of binary structures has bounded twin-width if and only if it can be expanded into a monadically dependent class by an oriented graph with bounded independence number.

 Moreover, \FO-model checking is {\FPT} on  such a monadically dependent expansion.
\end{theorem}
\begin{proof}
  Let $\sigma$ be a binary relational signature and let $\mathscr C$ be a class of $\sigma$-structures.
    We first prove the characterization of twin-width boundedness.
    
    Assume $\mathscr C$ has bounded twin-width. According to \zcref{thm:tww4}, $\mathscr C$ has a monadic expansion by a linear order, which is a transitive tournament, hence an oriented graph with independence number $1$.

    Conversely, assume that $\mathscr C$ can be expanded into a monadically dependent class $\mathscr D$ by an oriented graph with independence number strictly less than $k$ (the arcs corresponding to a new relation $R$). Note that $\mathscr D$ is a monadically dependent class of $(\sigma\cup\{R\})$-structures.   Let $\mathscr D^R=\{\mathbf M^R\colon \mathbf M\in\mathscr D\}$.
    According to \zcref{thm:alpha}, the class $\mathscr D^R$ has bounded twin-width. Hence, according to \zcref{thm:tww4}, the class $\mathscr D^R$ can be expanded into a monadically dependent class by a linear order $<$.  Let $\sigma^+$ be the signature $\sigma\cup\{R,<\}$.
    By accordingly expanding the original structures in $\mathscr D$ we get a class $\mathscr E$ of $\sigma^+$-structures such that
    $\mathscr E^{\sigma^+\setminus\{<\}}$ (i.e., $\mathscr D$) and $\mathscr E^{\{R,<\}}$ are monadically dependent. Hence, according to \zcref{thm:main}, $\mathscr E$ is monadically dependent.
   By \zcref{thm:tww4}, the class 
   $\mathscr E^{\sigma^+\setminus\{<\}}$ (i.e., $\mathscr D$)
has bounded twin-width,  and so does the class $\mathscr C$.
   
    We now consider the algorithmic complexity part of the statement: According to \cite{geniet2025orderlogictwinwidthtournaments}, the monadically dependent expansion of $\mathbf M^R\in\mathscr D^R$ by a linear order can be computed in polynomial time.
    According to the above, this provides a monadically dependent expansion $\mathscr E$ of $\mathscr D$ by a linear order. According to \cite[Theorem 9]{tww5}, \FO-model checking is {\FPT} on $\mathscr E$.
\end{proof}

As tournaments are oriented graphs of bounded independence number and linear orders are tournaments, we get the following corollary of \zcref{thm:tww4,thm:nalpha}.

\begin{corollary}[store=cor:tourn]
    A class of binary structures has bounded twin-width if and only if it can be expanded into a monadically dependent class by a tournament.
\end{corollary}

Also, as the independence number of the comparability graph of a poset $P$ is nothing but the width of $P$, we have
\begin{corollary}[store=cor:width]
    Let $k$ be a positive integer.
    A class of binary structures has bounded twin-width if and only if it can be expanded into a monadically dependent class by a poset of width at most $k$.
\end{corollary}

\subsection{New delineated classes of oriented graphs}
For a positive integer $k$, let $\mathscr A_k$ be the class of all graphs $G$ such that $G$ has an independent set $I$ with $\alpha(G-I)<k$, and let $\oak$ be the class of all the orientations of graphs in $\mathscr A_k$.

\begin{lemma}
\label{lem:PMIS}
    Assume $G$ has an independent set $I$ such that $\alpha(G-I)<k$.
    Then, an independent set with this property can be found in polynomial time.
\end{lemma}
\begin{proof}
    First notice that if such an independent set exists, we can assume that it is a maximal independent set.
    \begin{claim}
        Assume $\alpha(G-I)<k$. Then $G$ contains at most $n^{k-1}$ maximal independent sets.
    \end{claim}
    \begin{clproof}
    Let $J$ be a maximal independent set of $G$. Then, $J\setminus I$ is an independent set of $G-I$, hence
    $|J\setminus I|<k$.
    Let $J'\is (J\setminus I)\cup (I\setminus \bigcup_{v\in J\setminus I}N(v))$.
    As no neighbor of a vertex in $J$ is in $J$, we get $J\subseteq J'$. As $J'$ is independent by construction, we deduce $J=J'$ from the maximality of $J$.
    It follows that $J\mapsto J\setminus I$ is an injection from maximal independent sets of $G$ to subsets of less than $k$ vertices.
    \end{clproof}
    
   We enumerate all maximal independent sets of $G$ in $O(n^{k}m)$-time using \cite{tsukiyama1977new}, where $n$ is the number of vertices of $G$ and $m$ its number of edges.
    For each maximal independent set $J$, we check whether $\alpha(G-J)<k$ in time $n^{k\omega/3}$ using \cite{nevsetvril1985complexity}, where $\omega < 2.4$ is the fast matrix multiplication constant. Hence, the global computation time is $O(n^{k(1+\omega/3)}m)$, which is polynomial in the input size.
\end{proof}

\begin{theorem}[store=thm:delin]
    For every integer $k$, the class $\oak$ of all the orientations of graphs $G$ such that 
    $\min\{\alpha(G-I)\colon I\text{ independent set of }G\}<k$
     is delineated  by twin-width and \FO-model checking is {\FPT} on every monadically dependent subclass of $\oak$.
\end{theorem}
\begin{proof}
    For each $\mathbf G\in\oak$, let $I(\mathbf G)$ be an independent set of $\mathbf G$ such that $\alpha(\mathbf G-I(\mathbf G))<k$.

    Let $\oF$ be a monadically dependent subclass of $\oak$.
    We first consider the case where, for each $\mathbf G\in\oF$, no two vertices in $I(\mathbf G)$ have the same in- and out-neighbors.

    Let $\oD\is\{\mathbf G-I(\mathbf G)\colon \mathbf G\in\oF\}$. As $\oD$ is a transduction of $\oF$, $\oD$ is monadically dependent, hence has bounded twin-width according to \zcref{thm:alpha}. 
    It follows from \zcref{thm:tww4} that 
    there exists an expansion $\mathbf H\mapsto \widehat{\mathbf H}$ of oriented graphs $\mathbf H\in\oD$ by a linear order $<$ such that 
    $\{\widehat{\mathbf H}\colon \mathbf H\in\oD\}$ is monadically dependent.

    Let $\sigma=\{R,S,R'\}$. We define a mapping $\mathbf G\mapsto \widetilde{\mathbf G}$ of oriented graphs $\mathbf G\in\oF$ to $\sigma$-structures $\widetilde{\mathbf G}$ as follows:
    \begin{itemize}
        \item $\widetilde{\mathbf G}$ has the same domain as $\mathbf G$;
        \item $\widetilde{\mathbf G}\models R(u,v)$ if none of $u$ and $v$ belongs to $I(\mathbf G)$ and $\mathbf G\models R(u,v)$;
        \item $\widetilde{\mathbf G}\models S(u,v)$ if  $u$ or $v$ belongs to $I(\mathbf G)$ and $\mathbf G\models R(u,v)$;
        \item $\widetilde{\mathbf G}\models R'(u,v)$ if none of $u$ and $v$ belongs to $I(\mathbf G)$ and $\widehat{\mathbf H}\models (u<v)$, where $\mathbf H=G-I(\mathbf G)$.
    \end{itemize}
    Note that, by construction, $\widetilde{\mathbf G}^R$ and $\widetilde{\mathbf G}^{R'}$ both belong to $\obk$ and that 
    $\widetilde{\mathbf G}$ has the property $\Pi^k_{R,R'}$.

    Let $\oE=\{\widetilde{\mathbf G}\colon \mathbf G\in\oF\}$.
    The classes $\oE^{\sigma\setminus\{R'\}}$ (which is a transduction of $\oF$) and $\oE^{\{R,R'\}}$ (which is obtained from $\{\widehat{\mathbf H}\colon \mathbf H\in\oD\}$ by adding isolated vertices) are both monadically dependent. Thus, according to \zcref{thm:main}, the class $\oE$ is monadically dependent.

    For $\widetilde{\mathbf G}\in\oE$, if $I=I(\mathbf G)$, then $R'$ defines a linear order $<_0$ on the complement of $I$.
    We define a global linear order on $\widetilde{\mathbf G}\in\oE$ by $u<v$ if
\begin{itemize}
    \item either none of $u$ and $v$ are in $I$ and $u<_0 v$,
    \item or $u\notin I$ and $v\in I$,
    \item or $u$ and $v$ are both in $I$ and the minimum $z$ in $<_0$ which is not connected the same way to $u$ and $v$ is such that $z$ is not adjacent to $u$ or $z$ is adjacent to both $u$ and $v$ and an in-neighbor of $u$.
\end{itemize}
This is indeed a linear order, as $I$ does not contain two vertices with the same in- and out-neighbors. As this linear order can be obtained by transduction, we deduce that the class $\oE$ admits a monadically dependent expansion by a linear order hence, according to \zcref{thm:tww4}, has bounded twin-width. As $\oF$ is a transduction of $\oE$, the class $\oF$ has bounded twin-width as well.

As adding twins does not change the twin-width, the general case also follows.

Concerning the parametrized complexity for \FO-model checking: starting with $\mathbf G\in\oF$, we can find an admissible independent set $I(\mathbf G)$ in polynomial time according to \zcref{lem:PMIS}. The linear order $<_0$ on $\mathbf G-I$ can be computed in polynomial time \cite{geniet2025orderlogictwinwidthtournaments}. This order can be extended into $<$ in polynomial time (here we put twins in $I$ as consecutive vertices in an arbitrary order). The above proof ensures that the class of expanded oriented graphs has bounded 
twin-width. Then, \FO-model checking is performed using the algorithm presented in \cite{tww5}.
\end{proof}

\begin{corollary}[store=cor:split]
    The class of oriented split-graphs is delineated by twin-width and \FO-model checking is {\FPT} on every monadically dependent class of oriented split graphs.
\end{corollary}

It is easily checked that the class $\mathscr A_k$ excludes the matching $M_k$ as an induced subgraph. Hence the following problem.  
\begin{problem}
    Is the class of all orientations of $M_k$-free graphs delineated by twin-width?
\end{problem}

In \cite{geniet2025first}, it is proved that the class of circular arc graphs is delineated. Since proper circular arc graphs are exactly those graphs that can be oriented as local tournaments \cite{huang1995structure}, it is natural to ask whether the class of local tournaments is delineated.
We prove that it is indeed the case, providing another generalization of  \zcref{thm:tourn}.

We recall some terminology introduced in \cite{huang1995structure}. 
An oriented graph $\mathbf G$ is \emph{round-oriented} if its vertices can be circularly ordered $v_i$ ($i\in\mathbb Z_n$) so that, for each vertex $v_i$, there are non-negative integers $r_i$ and $\ell_i$ such that out-neighbors of $v_i$ are $v_{i+1},\dots,v_{i+r_i}$ and the in-neighbors of $v_i$ are $v_{i-1},\dots,v_{i-\ell_i}$. To \emph{substitute} an oriented graph $\mathbf S$ to a vertex $v$ of an oriented graph $\mathbf G$ means to form a new oriented graph $\mathbf G'$ by replacing $v$ with $\mathbf S$ so that in $\mathbf G'$ every vertex of $\mathbf S$ dominates every out-neighbor of $v$ and is dominated by every in-neighbor of $v$.
An arc $(u,v)$ of $\mathbf G$ is \emph{balanced} if $u$ and $v$ have the same closed neighborhood in $\Gaif(\mathbf G)$.
A local tournament $\mathbf G$ is \emph{reduced} if $\mathbf G$ has no balanced edge, that is if $\Gaif(\mathbf G)$ has no adjacent twins. A \emph{U-reversal} of an oriented graph $\mathbf G$ is an operation which reverses the directions of all the unbalanced arcs within one component of $\overline{\Gaif(\mathbf G)}$, or reverses the directions of all unbalanced arcs between two fixed components of 
$\overline{\Gaif(\mathbf G)}$.

\begin{etheorem}[{\cite[Theorem 1.2]{huang1995structure}}]
\label{thm:lt_struct}
    Let $\mathbf G$ be a connected oriented graph. Then $\mathbf G$ is a local tournament if and only if it is obtained from a reduced round-oriented graph
    $\mathbf R$ by substituting a tournament $\mathbf T_v$ to each vertex $v$ of $\mathbf R$ and then performing U-reversals. 
\end{etheorem}

\begin{lemma}
    \label{lem:lt_struct_nb}
   Let $\mathbf G$ be a connected oriented graph such that
   $\overline{\Gaif(\mathbf G)}$ is not bipartite.
   
   Then $\mathbf G$ is a local tournament if and only if 
   it is obtained from a reduced round-oriented graph
    $\mathbf R$ by substituting a tournament $\mathbf T_v$ to each vertex $v$ of $\mathbf R$.     
\end{lemma}
\begin{proof}
    According to  \zcref{thm:lt_struct}, $\mathbf G$ is obtained from a reduced round-oriented graph
    $\mathbf R$ by substituting a tournament $\mathbf T_v$ to each vertex $v$ of $\mathbf R$ and then performing U-reversals. As substituting  tournaments to vertices and performing U-reversals do not change the property of the complement of being bipartite, $\overline{\Gaif(\mathbf R)}$ is not bipartite. According to \cite[Theorem 4.5]{huang1995structure}, $\overline{\Gaif(\mathbf R)}$ is then connected and has exactly two orientations as a local tournament, obtained from each other by reversing all the arcs of $\mathbf R$. 
    As reversing all the edges in a reduced round-oriented graph produces a reduced round-oriented graph, we get that no U-reversal is needed by choosing the good reduced round-oriented graph.
\end{proof}

\begin{lemma}
\label{lem:approx}
    Let $\mathbf G$ be obtained from a reduced round-oriented graph $\mathbf R$ by substituting a tournament $\mathbf T_v$ 
    (with vertex set $T_v$) to each vertex $v$ of $\mathbf R$.

    For $x,y$ in $\mathbf G$ define $x\approx y$ if
    $x$ and $y$ have the same closed neighborhood in $\Gaif(\mathbf G)$.

    Then, the relation $\approx$ is an equivalence relation, and the equivalence class of $x\in T_v$ is $T_v$ and
    $\mathbf R$ is isomorphic to any induced oriented subgraph of $\mathbf G$ obtained by keeping exactly one vertex in each equivalence class.
\end{lemma}
\begin{proof}
    By construction, if $x$ and $y$ belong to the same set $T_v$, then $x\approx y$.
    Conversely, assume that $x\in T_u, y\in T_v$, and $x\approx y$.
    Let $N_u$ and $N_v$ be the closed neighborhood of $u$ and $v$ in $\Gaif(\mathbf R)$. By construction, the closed neighborhood of $x$ (resp.\ of $y$) in $\Gaif(\mathbf G)$ is
    $\bigcup_{w\in N_u} T_w$ (resp.\ $\bigcup_{w\in N_v} T_w$).
    As the sets $T_v$ form a partition of the vertex set of $\mathbf G$, we deduce from $x\approx y$ that $u$ and $v$ have the same closed neighborhood in $\Gaif(\mathbf R)$. As $\mathbf R$ is reduced, we deduce $u=v$.

    That $\mathbf R$ is isomorphic to any induced oriented subgraph of $\mathbf G$ obtained by keeping exactly one vertex in each equivalence class is direct from the construction of $\mathbf G$ from $\mathbf R$.
\end{proof}

\begin{lemma}
\label{lem:round}
The class of reduced round-oriented graphs is delineated by twin-width.
\end{lemma}
\begin{proof}
    Let $\oC$ be a monadically dependent class of reduced round-oriented graphs.
    By transduction, we transform $\mathbf G$ into $H$ as follows: each vertex  $v$ is copied into $(v,-1), (v,0)$, and $(v,1)$. We have
    $H\models E((u,i),(v,j))$ if $u=v$ and $(i,j)\notin\{(-1,1),(1,-1)\}$ or
    $u\neq v$ and
    \begin{itemize}
        \item $\mathbf G\models R(v,u)$ and ($i=-1$ or $j=1$),
        \item or $\mathbf G\models R(u,v)$ and ($i=1$ or $j=-1$),
        \item or $\mathbf G\models \exists w\ R(u,w)\wedge R(w,v)$ and $(i,j)=(1,-1)$,
        \item or $\mathbf G\models \exists w\ R(v,w)\wedge R(w,u)$ and $(i,j)=(-1,1)$.
    \end{itemize}
    (The construction of the circular arc graph $H$ from a reduced round-oriented graph is shown in \zcref{fig:circ}.)
  
    The graph $H$ is a circular arc graph.
    Let $\mathscr D$ be the class of the so-obtained graphs $H$.
    As $\mathscr D$ can be transduced from $\oC$, it is monadically dependent. As the class of circular arc graphs is delineated by twin-width \cite{geniet2025first}, $\mathscr D$ has bounded twin-width.

    To go back, we color all the vertices $(v,i)$ by color $i$
    and $i\in\{-1,0,1\}$ we denote by $P_i(x)$ the property that $x$ is colored $i$. We keep only the vertices colored $0$. Note that for every vertex $(v,0)$ the vertex $(v,1)$ is the only neighbor of $(v,0)$ colored $1$ that is not a neighbor of all the neighbors of $(v,0)$. 
    From this property, we derive that we can let $R(x,y)$ if
\[
(x\neq y)\wedge \exists y',y'' \bigl( P_1(y')\wedge P_{-1}(y'')\wedge E(y,y')\wedge E(y,y'')\wedge \neg E(y',y'')\wedge E(x,y'')\bigr).
\]
    This proves that $\oC$ is a transduction of $\mathscr D$.
    Hence, $\oC$ has bounded twin-width.
\end{proof}

\begin{figure}[ht]
    \centering
    \includegraphics[width=\linewidth]{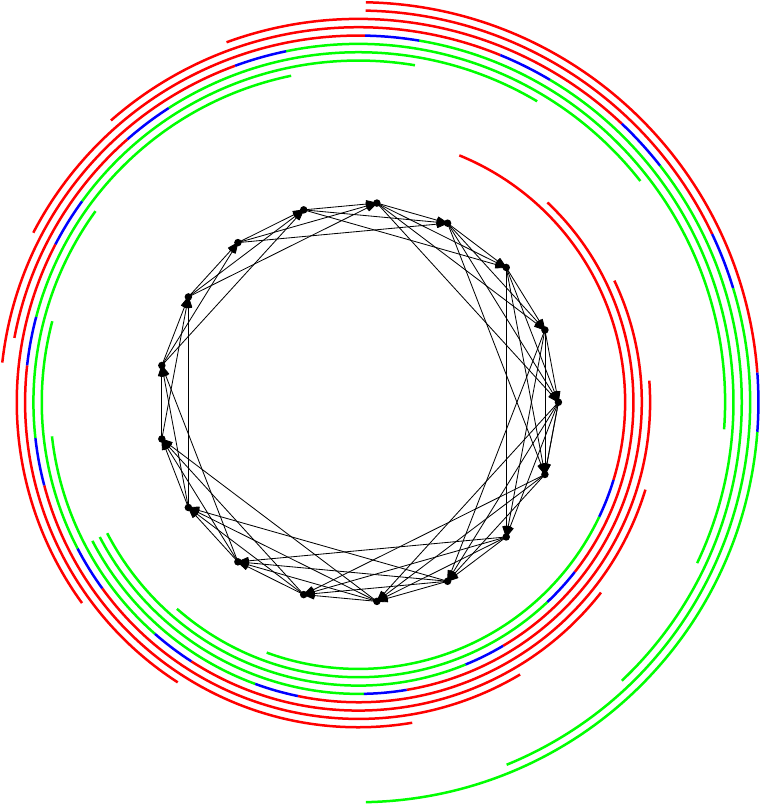}
    \caption{The circular arc graph constructed from reduced round-oriented graph (in the middle). The arcs shown on a same circle correspond to a same vertex $u$ ($(u,-1)$ is red, $(u,0)$ is blue, and $(u,1)$ is green.)}
    \label{fig:circ}
\end{figure}

\begin{theorem}[store=thm:ltourn]
    The class of local tournaments is delineated by twin-width. In other words,
    a class of local tournaments has bounded twin-width if and only if it is monadically dependent.

    Moreover, if a class $\mathscr C$ of local tournaments has bounded twin-width, then \FO-model checking is {\FPT} on $\mathscr C$.
\end{theorem}
\begin{proof}
    Let $\oC$ be a monadically dependent class of local tournaments.
    We split this class into $\oC_1$ and $\oC_2$, where 
    $\mathbf G\in\oC$ belongs to $\oC_1$ if $\overline{\Gaif(\mathbf G)}$
    is bipartite, and to $\oC_2$ otherwise.
    Note that we can decide whether $\mathbf G\in\oC_1$ or $\mathbf G\in\oC_2$ in linear time.
    Moreover, in order to prove that the class of all local tournaments is delineated by twin-width, it will be sufficient to prove that $\oC_1$ and $\oC_2$ both have bounded twin-width.
    
    As every $\mathbf G\in\oC_1$ has independence number at most $2$, $\oC_1$ has bounded twin-width and \FO-model checking is {\FPT} for oriented graphs in $\oC_1$
    according to \zcref{thm:alpha}.
    Thus, we may restrict our attention to input graphs in the class $\oC_2$.

    According to \zcref{lem:lt_struct_nb}, each $\mathbf G\in\oC_2$  is obtained from a reduced round-oriented graph
    $\Red(\mathbf G)$ by substituting a tournament $\mathbf T_v$ to each vertex $v$ of $\Red(\mathbf G)$. Moreover, the partition into sets $T_v$ and the reduced round-oriented graph $\Red(\mathbf G)$ can be computed in quadratic time using the equivalence relation $\approx$ introduced in 
    \zcref{lem:approx}.

    We denote by $\overrightarrow{\mathscr T}$
    (resp.\ $\overrightarrow{\mathscr R}$) the class of all the tournaments $\mathbf T_v$ (resp.\ of all the reduced round-oriented graphs $\Red(\mathbf G)$) derived from oriented graphs $\mathbf G\in\oC_2$.
    As all the tournaments $\mathbf T_v=\mathbf G[T_v]$ and 
    the reduced round-oriented graphs $\Red(\mathbf G)$ derived from some $\mathbf G\in\oC_2$ are induced subgraphs of $\mathbf G$, both $\overrightarrow{\mathscr T}$
    and $\overrightarrow{\mathscr R}$ are transductions of $\oC_2$, hence are monadically dependent.

    According to \zcref{thm:tourn}, the class of all tournaments is delineated by twin-width, hence $\overrightarrow{\mathscr T}$ has bounded twin-width. It follows from \cite{geniet2023first,geniet2025orderlogictwinwidthtournaments} that there is a polynomial time algorithm that computes, for an input tournament $\mathbf T\in\overrightarrow{\mathscr T}$, a linear order $L_1(\mathbf T)$ such that the class of all the tournaments in $\overrightarrow{\mathscr T}$ expanded by the computed linear orders has bounded twin-width.

    On the other hand, according to \zcref{lem:round}, the class of all reduced round-oriented graphs is delineated by twin-width, hence $\overrightarrow{\mathscr R}$ has bounded twin-width.
    As the oriented graphs in $\oC_2$ are obtained from graphs in the class $\overrightarrow{\mathscr R}$ (which has bounded twin-width) by substituting graphs from $\overrightarrow{\mathscr T}$ (which has bounded twin-width) to vertices, the class $\oC_2$ has bounded twin-width.
    
    We now consider the parametrized complexity aspect. 
    As noticed above, the reduced round-oriented graph $\Red(\mathbf G)$ associated to $\mathbf G\in\oC_2$ can be computed in quadratic time. We compute the circular arc graph $H(\Red(\mathbf G))$ associated to $\Red(\mathbf G)$ in the proof of \zcref{lem:round} in polynomial time. Let $\mathscr A$ be the corresponding class of circular arc graphs. As $\mathscr A$ is a transduction of $\overrightarrow{\mathscr R}$, the class $\mathscr A$ has bounded twin-width. It follows from \cite{geniet2025first} that we can compute in polynomial time, for each circular arc graph $H\in\mathscr A$, a linear order $L_2(H)$, in such a way that the class of all the graphs $H\in\mathscr A$ expanded by the computed linear orders has bounded twin-width.
    As $\Red(\mathbf G)$ can be obtained as a simple interpretation of $H(\Red(\mathbf G))$, the class obtained from $\overrightarrow{\mathscr R}$ by expanding each $\Red(\mathbf G)$ by (the restriction of) the linear order $L_2(H(\Red(\mathbf G)))$ has bounded twin-width.

    As substitution of binary structures from a class with bounded twin-width to the vertices of binary structures from a class with bounded twin-width produces a class with bounded twin-width, we derive that the class obtained by substituting in the expansion of $\Red(\mathbf G)$ ordered tournaments from the expansion of $\overrightarrow{\mathscr T}$ to vertices has bounded twin-width. This substitution, which can easily be computed in polynomial time, is an expansion $\mathbf G^+$ of $\mathbf G$ by a linear order. Then, the fixed parameter tractability of \FO-model checking follows from the fixed parameter tractability of \FO-model checking for monadically dependent classes of ordered binary structures~\cite{tww5}.
\end{proof}

Naturally, one could wonder whether an analog of \zcref{thm:nalpha} could hold, where oriented graph with bounded independence number is replaced by local tournament.
The answer is negative, as circuits are local tournaments and every class of cubic graphs can be expanded into a monadically dependent class by a circuit.

\bibliographystyle{amsplain}
\bibliography{ref}
\end{document}